\documentclass[12pt,draftcls,journal,letterpaper,twosides,onecolumn]{IEEEtran}

\usepackage{amsmath,amssymb,amscd,latexsym,dsfont,wasysym,bbm}
\usepackage{float}
\usepackage{color}
\usepackage{graphicx,subfigure}
\usepackage{multirow,multicol}
\usepackage{verbatim}
\usepackage{psfrag}
\usepackage{comment}
\usepackage{makeidx}
\usepackage[]{algorithm}
\usepackage{algorithmic}
\usepackage{cite}
\usepackage{cases}
\usepackage{amsthm}
 \usepackage{setspace}
\usepackage[dvipsnames]{xcolor}
\newcommand{\tr}[1]{\textrm{#1}}
\newcommand{\mr}[1]{\mathrm{#1}}
\newcommand{\tnr}[1]{{\textnormal{#1}}}

\newcommand{\mc}[1]{\mathcal{#1}}
\newcommand{\mf}[1]{\mathsf{#1}}

\newcommand{\ms}[1]{\mathds{#1}}

\newcommand{\ov}[1]{\overline{#1}}

\newcommand{\bc}{\boldsymbol{c}}

\newcommand{\bp}{\boldsymbol{p}}

\newcommand{\bQ}{\boldsymbol{Q}}

\newcommand{\bx}{\boldsymbol{x}}

\newcommand{\by}{\boldsymbol{y}}

\newcommand{\bz}{\boldsymbol{z}}

\newcommand{\figref}[1]{Fig.~\ref{#1}}

\newcommand{\secref}[1]{Sec.~\ref{#1}}
\newcommand{\appref}[1]{Appendix~\ref{#1}}

\newcommand{\exref}[1]{Example~\ref{#1}}

\newcommand{\lemref}[1]{Lemma~\ref{#1}}

\newcommand{\e}{\mr{e}}

\newcommand{\ie}{i.e.,~} 		
\newcommand{\eg}{e.g.,~}	
\newcommand{\cf}{cf.~}		

\newcommand{\argmax}{\mathop{\mr{argmax}}}

\newcommand{\set}[1]{\{#1\}}
\newcommand{\SET}[1]{\left\{#1\right\}}

\newcommand{\cd}{\cdot}
\newcommand{\ld}{\ldots}

\newcommand{\PR}[1]{\Pr\SET{#1}}       	
\newcommand{\pdf}{p}            			

\newcommand{\cdf}{F}            			
\newcommand{\ccdf}{F^\tnr{c}}            			

\newcommand{\IND}[1]{\ms{I}\big[{#1}\big]}   	
\newcommand{\Ex}{\ms{E}}     			

\newcommand{\mcA}{\mc{A}}

\newcommand{\mcF}{\mc{F}}

\newcommand{\mcM}{\mc{M}}

\newcommand{\mcP}{\mc{P}}

\newcommand{\mcR}{\mc{R}}
\newcommand{\mcS}{\mc{S}}

\newcommand{\mcX}{\mc{X}}

\newcommand{\mfa}{\mf{a}}

\newcommand{\mfk}{\mf{k}}

\newcommand{\mfm}{\mf{m}}

\newcommand{\mfs}{\mf{s}}

\newcommand{\mfM}{\mf{M}}

\newcommand{\mfI}{\mf{I}}

\newcommand{\mfR}{\mf{R}}

\newcommand{\Real}{\mathbb{R}}		
\newcommand{\Binary}{\mathbb{B}}	

\newcommand{\SNR}{\mathsf{snr}}  
\newcommand{\SNRrv}{\mathsf{SNR}}  
\newcommand{\SNRav}{\ov{\mathsf{snr}}}  
\newcommand{\X}{\mcX}	

\newcommand{\R}{R}              	
\newcommand{\Nb}{{\mathop{N_\tnr{b}}}} 	
\newcommand{\Ns}{{\mathop{N_\tnr{s}}}} 	
\newcommand{\Nc}{N_\tnr{c}} 	

\newcommand{\IR}{\tnr{ir}}

\newcommand{\ack}{\mf{ACK}}  
\newcommand{\nack}{\mf{NACK}}  
\newcommand{\kmax}{K}		

\newcommand{\xp}{\tnr{xp}}

\newcommand{\Isig}{I^\Sigma}  
\newcommand{\Rsig}{\R^\Sigma}  

\newcommand{\Eharqa}{\mc{E}_{\ack}}
\newcommand{\Eharqn}{\mc{E}_{\nack}}

\newcommand{\dBval}{\tnr{dB}} 
\newcommand{\dB}{~\tnr{[dB]}} 

\usepackage[acronym,nonumberlist]{glossaries} 
\usepackage{hyperref}
\newacronym[\glsshortpluralkey=PDFs,\glslongpluralkey=probability density functions]{pdf}{PDF}{probability density function}
\newacronym[\glsshortpluralkey=CDFs,\glslongpluralkey=cumulative density functions]{cdf}{CDF}{cumulative density function}
\newacronym[\glsshortpluralkey=CCDFs,\glslongpluralkey=complementary cumulative density functions]{ccdf}{CDF}{complementary cumulative density function}
\newacronym[\glsshortpluralkey=PMFs,\glslongpluralkey=probability mass functions]{pmf}{PMF}{probability mass function}
\newacronym[]{lhs}{l.h.s.}{left-hand side}
\newacronym[]{rhs}{r.h.s.}{right-hand side} 

\newacronym[]{bicm}{BICM}{bit-interleaved coded modulation}
\newacronym[]{bicmid}{BICM-ID}{BICM with iterative demapping}
\newacronym[]{cm}{CM}{coded modulation}
\newacronym[]{tcm}{TCM}{trellis-coded modulation}
\newacronym[]{mlc}{MLC}{multi-level coding}
\newacronym[]{pam}{PAM}{pulse amplitude modulation}
\newacronym[]{bpsk}{BPSK}{binary phase shift keying}
\newacronym[]{qam}{QAM}{quadrature amplitude modulation}
\newacronym[]{psk}{PSK}{phase shift keying}
\newacronym[\glsshortpluralkey=LLRs,\glslongpluralkey=logarithmic likelihood ratios]{llr}{LLR}{logarithmic likelihood ratio}
\newacronym[]{map}{MAP}{maximum a posteriori}
\newacronym[]{ml}{ML}{maximum likelihood}
\newacronym[\glsshortpluralkey=MIs,\glslongpluralkey=mutual informations]{mi}{MI}{mutual information}
\newacronym[\glsshortpluralkey=GMIs,\glslongpluralkey=generalized mutual informations]{gmi}{GMI}{generalized mutual information}
\newacronym[]{eesm}{EESM}{exponential effective-SNR-mapping}
\newacronym[]{bicm-gmi}{BICM-GMI}{BICM generalized mutual information}
\newacronym[]{awgn}{AWGN}{additive white Gaussian noise}
\newacronym[]{amc}{AMC}{adaptive modulation and coding}
\newacronym[]{csi}{CSI}{channel state information}
\newacronym[]{cqi}{CQI}{channel quality indicator}
\newacronym[]{sp}{SP}{set-partitioning}
\newacronym[]{gsm}{GSM}{global system for mobile communications}
\newacronym[]{edge}{EDGE}{enhanced data rates for GSM evolution}
\newacronym[]{3gpp}{3GPP}{3rd generation partnership project}
\newacronym[]{lte}{LTE}{Long Term Evolution}
\newacronym[]{dvb}{DVB}{digital video broadcasting}

\newacronym[\glsshortpluralkey=CCs,\glslongpluralkey=convolutional codes]{cc}{CC}{convolutional code}
\newacronym[\glsshortpluralkey=PCCCs,\glslongpluralkey=parallel concatenated convolutional codes]{pccc}{PCCC}{parallel concatenated convolutional code}
\newacronym[\glsshortpluralkey=TCs,\glslongpluralkey=turbo codes]{tc}{TC}{turbo code}
\newacronym{ldpc}{LDPC}{low-density parity-check}
\newacronym[]{ofdm}{OFDM}{orthogonal frequency-division multiplexing}
\newacronym[]{bep}{BEP}{bit-error probability}
\newacronym[]{wep}{WEP}{word-error probability}
\newacronym[]{sep}{SEP}{symbol-error probability}
\newacronym[]{pep}{PEP}{pairwise-error probability}
\newacronym[]{ttcm}{TTCM}{turbo-trellis coded modulation}
\newacronym[]{uep}{UEP}{unequal error protection}
\newacronym[\glsshortpluralkey=CENCs,\glslongpluralkey=convolutional encoders]{cenc}{CENC}{convolutional encoder}
\newacronym[]{mimo}{MIMO}{multiple-input multiple-output}
\newacronym[\glsshortpluralkey=SNRs,\glslongpluralkey=signal-to-noise ratios]{snr}{SNR}{signal-to-noise ratio}
\newacronym[\glsshortpluralkey=SINRs,\glslongpluralkey=the signal-to-interference-plus-noise ratios]{sinr}{SINR}{the signal-to-interference-plus-noise ratio}
\newacronym[]{msb}{MSB}{most-significative bit}
\newacronym[]{bcjr}{BCJR}{Bahl--Cocke--Jelinek--Raviv}
\newacronym[\glsshortpluralkey=SEDs,\glslongpluralkey=squared Euclidean distances]{sed}{SED}{squared Euclidean distance}
\newacronym[\glsshortpluralkey=EDs,\glslongpluralkey=Euclidean distances]{ed}{ED}{Euclidean distance}
\newacronym[\glsshortpluralkey=MEDs,\glslongpluralkey=minimum Euclidean distances]{med}{MED}{minimum Euclidean distance}
\newacronym[]{core}{CoRe}{constellation rearrangement}
\newacronym[]{msd}{MSD}{multistage decoding}
\newacronym[]{pdl}{PDL}{parallel decoding of the individual levels}
\newacronym[\glsshortpluralkey=GCs,\glslongpluralkey=Gray codes]{gc}{GC}{Gray code}
\newacronym[]{brgc}{BRGC}{binary-reflected Gray code}
\newacronym[]{nbc}{NBC}{natural binary code}
\newacronym[]{fbc}{FBC}{folded-binary code}
\newacronym[]{bsgc}{BSGC}{binary semi-Gray code}
\newacronym[]{msp}{MSP}{modified set-partitioning}
\newacronym[]{ssp}{SSP}{semi set-partitioning}
\newacronym[]{fhd}{FHD}{free Hamming distance}
\newacronym[]{mfhd}{MFHD}{maximum free Hamming distance}
\newacronym[]{ods}{ODS}{optimal distance spectrum}
\newacronym[]{iud}{i.u.d.}{independent and uniformly distributed}
\newacronym[]{ud}{u.d.}{uniformly distributed}
\newacronym[]{iid}{i.i.d.}{independent, identically distributed}
\newacronym[]{bico}{BICO}{binary-input continuous-output}
\newacronym[]{gh}{GH}{Gauss--Hermite}

\newacronym[\glsshortpluralkey=BSs,\glslongpluralkey=base-stations]{bs}{BS}{base-station}
\newacronym[\glsshortpluralkey=MSs,\glslongpluralkey=mobile-stations]{ms}{MS}{mobile-stations}

\newacronym[]{phy}{PHY}{physical layer} 
\newacronym[]{llc}{LLC}{logical link control} 
\newacronym[]{mac}{MAC}{media access control} 
\newacronym[]{fft}{FFT}{fast Fourier transform} 
\newacronym[]{cf}{CF}{characteristic function} 
\newacronym[]{mgf}{MGF}{moment generating function} 
\newacronym[]{ee}{EE}{energy efficiency} 
\newacronym[]{kkt}{KKT}{Karush--Kuhn--Tucker} 
\newacronym[]{mcs}{MCS}{modulation/coding scheme} 
\newacronym[]{fec}{FEC}{forward error correction}
\newacronym[]{arq}{ARQ}{automatic repeat request}
\newacronym[]{harq}{HARQ}{hybrid ARQ}
\newacronym[]{tarq}{TARQ}{truncated HARQ}
\newacronym[]{rrharq}{RR-HARQ}{repetition redundancy HARQ}
\newacronym[]{irharq}{IR-HARQ}{incremental redundancy HARQ}
\newacronym[]{ack}{ACK}{positive acknowledgment}
\newacronym[]{nack}{NACK}{negative acknowledgment}
\newacronym[]{crc}{CRC}{cyclic redundancy check}
\newacronym[]{dp}{DP}{dynamic programming}
\newacronym[]{gp}{GP}{geometric programming}
\newacronym[]{per}{PER}{packet error rate}
\newacronym[]{op}{OP}{outage probability}
\newacronym[]{spa}{SPA}{saddle-point approximation}
\newacronym[]{mrc}{MRC}{maximum ratio combining}
\newacronym[]{mdp}{MDP}{Markov decision process}
\newacronym[]{pomdp}{POMDP}{Partially observable MDP}
\newacronym[]{psimdp}{PSI-MDP}{Partial State Information Markov Decision Process}

\newacronym[]{mm}{MM-HARQ}{multi-message HARQ}
\newacronym[]{xp}{XP-HARQ}{Cross-packet HARQ}
\newacronym[]{ts}{TS}{time-sharing}
\newacronym[]{sc}{SC}{superposition coding}
\newacronym[]{sbrq}{SBRQ}{systematic backward retransmission}
\newacronym[]{brq}{BRQ}{backward retransmission}
\newacronym[]{lharq}{L-HARQ}{layer-coded HARQ}
\newacronym[]{anlharq}{AoNL-HARQ}{all-or-none layer-coded HARQ}

\newacronym[]{ppp}{PPP}{Poisson point process}

\usepackage{tikz,pgfplots}
\pgfplotsset{width=3.4in,height=2.3in} 
\usetikzlibrary{positioning,shadows,shapes,fit,arrows,backgrounds}

\usepackage{epsf}           
\usepackage{graphicx}       
\usepackage{epsfig}         
\usepackage{pstricks}
\usepackage{pstricks-add}
\usepackage{pst-plot}
\usepackage{dsfont}

\usepackage{balance}

\newtheorem{lemma}{Lemma}

\newtheorem{example}{Example}

\begin{document}

\title{Adaptive Cross-Packet HARQ}
\author{Mohammed Jabi, Abdellatif Benyouss, Ma\"el Le Treust,  \'Etienne Pierre-Doray, and  Leszek Szczecinski
\thanks{%
M. Jabi, A. Benyouss, and L. Szczecinski are with 
INRS-EMT, University of Quebec, Montreal, Canada. e-mail: \{benyouss,jabi,leszek\}@emt.inrs.ca, }
\thanks{%
E. Pierre-Doray is with Polytechnique de Montr\'eal, Canada; he was also with INRS-EMT when this work was carried out. e-mail: \{etipdoray@gmail.com\}
}
\thanks{%
M.~Le Treust is with ETIS - UMR 8051 / ENSEA - Universit\'e de Cergy-Pontoise - CNRS, France. e-mail: \{mael.le-treust@ensea.fr.\}}
\thanks{%
Part of this work was presented at the IEEE Wireless Commun. Network. Conf. (WCNC), Doha, Qatar, April 2016.} %
}%


\maketitle
\thispagestyle{empty}

\begin{abstract}
In this work, we investigate a coding strategy devised to increase the throughput in \gls{harq} transmission over block fading channel. In our approach, the transmitter jointly encodes a variable number of bits for each round of \gls{harq}. The parameters (rates) of this joint coding can vary and may be based on the \gls{nack} provided by the receiver or, on the past (outdated) information about the channel states. These new degrees of freedom allow us to improve the match between the codebook and the channel states experienced by the receiver. The results indicate that significant gains can be obtained using the proposed coding strategy, particularly notable  when the conventional \gls{harq} fails to offer throughput improvement even if the number of transmission rounds is increased. The new cross-packet \gls{harq} is also implemented using turbo codes where we show that the theoretically predicted throughput gains materialize in practice, and we discuss the implementation challenges. 
\end{abstract}


\section{Introduction}\label{Sec:Intro}
In this work, in order to improve the throughput of the \gls{harq} transmission over block-fading channel,  we propose to use joint coding of multiple information packets into the same channel block and we develop methods to optimize the coding rates.  

\gls{harq} is used in modern communications systems to deal with unpredictable changes in the channel (due to fading), and with the distortion of the transmitted signals (due to noise). \gls{harq} relies on the feedback/acknowledgement channel, which is used by the receiver to inform the transmitter about the decoding errors (via \gls{nack}) and about the decoding success, via \gls{ack}. After \gls{nack}, the transmitter makes another  transmission \emph{round} which conveys additional information necessary to decode the packet. This continues till \gls{ack} is receiver and then a new \gls{harq} \emph{cycle} starts again for another information packet. In so-called \emph{truncated} \gls{harq}, the cycle stops also if the maximum number of rounds is attained. 

As in many previous works, \eg \cite{Caire01,Larsson14}, we will consider throughput as a performance measure assuming that residual errors are taken care of by the upper layers \cite{Jabi16}. We consider here the ``canonical'' problem defined in \cite{Caire01}, where the \gls{csi} is available at the receiver but not at the transmitter, which knows only its statistical description. The essential part of \gls{harq} is channel coding, which is done over many channel blocks as long as \gls{nack}s are obtained over the feedback channel. 

It was shown in \cite{Caire01} that \gls{harq}'s throughput  may approach the  ergodic capacity of the channel with sufficiently high ``nominal'' coding rate per round. However, such an approach is based on large number of \gls{harq} rounds, and thus has a limited practical value: long buffers are required which becomes a limiting factor for implementation of \gls{harq} \cite{Lee15}. 

On the other hand, using finite nominal coding rate and truncated \gls{harq}, the difference between the  throughput achievable using \gls{harq} and the theoretical limits may be large, especially, when we target throughput close to the nominal rate \cite{Larsson14}, \cite{Jabi15b}. 

To address this problem, various adaptive versions of \gls{harq} were proposed in the literature. For example, \cite{Cheng03,Visotsky03,LiuR03,Visotsky05,Kim11,Szczecinski13,Pfletschinger14} suggested to vary the length of the codewords so as to strike the balance between the number of channel uses and the chances of successful decoding. Their obvious drawback is that the resources assigned to the various \gls{harq} rounds are not constant which may leave an ``empty'' space within the block. 

To deal with this issue, it was proposed to share the block resources (power, time or bandwidth) between various packets in \eg  \cite{Zhang09,Takahashi10,Chaitanya11b,Steiner08,Jabi16}, to encode many packets into predefined size blocks as done in \cite{ElAoun10,AOUN12}, or to group variable-length codewords to fill the channel block \cite{Wang07,Szczecinski13}. A simplified approach was also proposed in \cite{Popovski14} to transmit the redundancy using two-step encoding.

These approaches implicitly implement  a joint coding of many packets into a single channel block. Here, we want to address the issue of cross-packet coding explicitly. The idea of this \gls{xp} is to get rid of the restricting assumptions proper to various heuristics developed before and to use a generic joint \gls{harq} encoder accepting many information packets and encoding them into a common codeword which fills the channel block. 


The contributions of this work are the following:
\begin{itemize}
\item We propose a general framework to analyze joint encoding of multiple packets which allow us to derive the relationship between the coding rates and the throughput. Our approach to cross-packet coding is similar to the one shown in  \cite{Hausl07,Chui07,Duyck10,Trillingsgaard14}, which, however, did not optimize the coding parameters. The optimization was proposed in \cite{Nguyen15}, however, due to complex decoding rules, it was very tedious and thus limited to the case of a simple channel model. In our work we simplify the problem assuming asymptotically long codewords are used, which leads to a compact description of the decoding criteria and allows us to solve the rate-optimization problem.
\item We consider the so-called multi-bit feedback to adapt the  coding rates to the channel state experienced by the receiver in the past transmission rounds of \gls{harq}. The same idea was exploited already \eg in\cite{Tuninetti11b,Nguyen12,Szczecinski13,Cheng03,Karmokar04,Jabi15,Visotsky05,Djonin08,Gopalakrishnan08b,Pfletschinger14,Jabi16}. The assumption of multi-bit feedback  not only simplifies the optimization but also yields the results which may be treated as the ultimate performance limits of any adaptation schemes when the instantaneous \gls{csi} is not available at the transmitter.
\item We optimize the coding rates  using the \gls{mdp} formulation \cite[Chap.~4]{Bertsekas07_book}, and compare the  proposed, \gls{xp} to the conventional \gls{irharq} from the perspective of attainable throughput. For the particular case of two transmission round, we obtain the optimal solution in closed-form.
\item 
We also present an analytical formula for attainable throughput using heuristic rate-adaptation inspired by the numerical results and which presents a notable gain over the conventional \gls{irharq}.
\item To obtain an insight into the practical constraints on the system design, we also show the results obtained when a turbo coding is adopted.
\end{itemize}

The remainder of the paper is organized as follows. We define the transmission model as well as the basic performance metrics in \secref{Sec:Model}. The idea of cross-packet coding  is explained in \secref{Sec:joint.codec}. The optimization of the rates in the proposed coding strategy is presented in \secref{Sec:Rate_Optimazation}. We discuss the effects of using a practical encoding/decoding schemes in \secref{Sec:Practical}. The numerical results are presented in form of short examples throughout the work to illustrate the main ideas. Conclusions are presented in \secref{Sec:Conclusions}. The optimization methods used to obtain the numerical results and the proof of decoding conditions are presented in appendices.

\section{Channel model and HARQ}\label{Sec:Model}

We consider a point-to-point  \gls{irharq} transmission of a packet $\mfm$ over a block fading channel. After each transmission, using a feedback/acknowledgement channel, the receiver tells the transmitter whether the decoding of $\mfm$ succeeded (\gls{ack}) or failed (\gls{nack}). We thus assume that error detection is possible (\eg via \gls{crc} mechanisms) and that the feedback channel is error-free. For simplicity, we ignore any loss of resources due to the \gls{crc} and the acknowledgement feedback.

The transmission of a single packet may thus require many transmission rounds which continue till the $\kmax$th round is reached or till \gls{ack} is received. When $\kmax$ is finite, we say that \gls{harq} is \emph{truncated}, otherwise we say it is \emph{persistent}. We define a \gls{harq} \emph{cycle} as the sequence of transmission rounds of the same packet $\mfm$. 

The received signal in the $k$th round is given by
\begin{align}\label{y.x.z}
\by_{k}=\sqrt{\SNR_{k}}\bx_{k}+\bz_{k},\quad k=1,\ld, \kmax
\end{align}
where $\bz_{k}$ and $\bx_k$ modelling, respectively, the noise and the transmitted codeword are $\Ns$-dimensional vectors, each containing  \gls{iid} zero mean, unit-variance random variables;  $\SNR_{k}$ is thus  the \gls{snr} at the receiver. The elements of $\bz_k$ are drawn from complex Gaussian distribution, and elements of $\bx_{k}$ -- from the uniform distribution over the set (constellation) $\mcX$.

During the $k$th round, $\SNR_{k}$ is assumed to be perfectly known/estimated at the receiver and unknown at the transmitter; it varies from one round to another and we model $\SNR_{k}, k=1,\ld,\kmax$ as the \gls{iid} random variables $\SNRrv$ with distribution $\pdf_{\SNRrv}(\SNR)$. 

\subsection{Conventional \gls{harq}}\label{Sec:HARQ}

In the conventional \gls{irharq}, a packet $\mfm \in\set{0,1}^{\R\Ns}$ is firstly encoded into a codeword $\bx=\Phi[\mfm]\in\X^{K\Ns}$ composed of $K\Ns$ complex symbols taken from a constellation $\X$ where $\Phi[\cd]$ is the coding function and $\R$ denotes the \emph{nominal} coding rate per block.\footnote{We clearly define the nominal rate as the coding rate per channel block because \gls{harq} is a variable-rate transmission: the number of used channel blocks is random, and the final transmission rate is random as well.} Then, the codeword $\bx$ is divided into $K$ disjoint subcodewords $\bx_k$ composed of different symbols \ie $\bx=[\bx_1,\bx_2,\ld,\bx_\kmax]$. After each round $k$, the receiver try to decode the packet $\mfm$ concatenating all received channel outcomes till the $k$th block
\begin{align}\label{yk.x.z}
 \by_{[k]}=[\by_{1},\ld,\by_{k-1},\by_{k}].
\end{align}

Following \cite{Caire01,Nguyen12}, we assume $\Ns$ large enough to make the random coding limits valid. Then, knowing the \gls{mi} $I_k=\mfI(X_k;Y_k|\SNR_k)$ between the random variables $X_k$ and $Y_k$  modeling respectively, the channel input and output in the $k$th block, allows us to determine when the decoding is successful or not: the decoding failure occurs in the $k$th round  if the accumulated \gls{mi} at the receiver is smaller than the coding rate
\begin{align}\label{nack.def}
\nack_k&\triangleq\set{ \big(I_1<R\big) \wedge \big(\Isig_2<R\big) \wedge \ld\wedge \big(\Isig_k<R\big) }\\
&= \SET{\Isig_k<R},
\end{align}
where $\Isig_k\triangleq\sum_{l=1}^k I_l$ is the \gls{mi} accumulated in $k$ rounds. Of course, the \gls{mi} depends on the \gls{snr}, \ie $I_k\equiv I_k(\SNR_k)$.
 
\gls{irharq} can be modelled as a Markov chain where the transmission rounds correspond to the states, and the \gls{harq} cycle corresponds to a renewal cycle in the chain. 
Thus, the long-term average throughput, defined as the average number of correctly received bits per transmitted symbol, may be calculated from the renewal-reward theorem: it is  a ratio between the average reward (number of bits successfully decoded per cycle) and the average renewal time (the expected number of transmissions needed to deliver the packet with up to $\kmax$ transmission rounds) \cite{Caire01}.

Let $f_k\triangleq\PR{\nack_k}, k\ge1$ be the probability of $k$ successive errors so the probability of successful decoding in the $k$th round is given by $\PR{\nack_{k-1}\wedge \Isig_k\ge\R }=f_{k-1}-f_k$ \cite{Caire01}. The throughput is then calculated as follows \cite{Caire01}
\begin{align}
\label{eta.rr.detail}
\eta^{\IR}_{\kmax}
&=\frac{R(1-f_1)+ R(f_1-f_2)+ \ld+R(f_{K-1}-f_K)}
{1\cd(1-f_1)+2\cd(f_1-f_2)+\ld+K\cd(f_{K-1})}\\
\label{eta.rr}
&=\frac{\R (1-f_K)}{1+\sum_{k=1}^{K-1}f_k}.
\end{align}



Because the instantaneous \gls{csi} is not available at the transmitter, the highest achievable throughput is given by the ergodic capacity\footnote{We use the term ``capacity'' to denote the achievable rate for a \emph{given} distribution of $X$.} of the channel \cite{Caire01,Wu10}
\begin{align}\label{C.erg}
\ov{C}\triangleq \Ex_\SNRrv[ I(\SNRrv]).
\end{align} 

However, achieving $\ov{C}$ is not obvious: as shown in \cite{Caire01}, it can be done growing simultaneously $\R$ and $\kmax$ to infinity but this approach is impractical due to large memory requirements.

\begin{example}[Two-states channel]\label{Ex:2states}
Consider a block-fading channel where the \gls{mi} can only take two values, $I_\tnr{a}$ and $I_\tnr{b}$, where $\PR{I=I_\tnr{a}}=1-p$ and $\PR{I=I_\tnr{b}}=p$. The ergodic capacity is given by $\ov{C}=I_\tnr{a}(1-p)+I_\tnr{b} p$. We force the \gls{harq} to deliver the packet at most in the last transmission, \ie $f_K=0$, which means that we impose the constraints on the coding rate $\R\leq KI_\tnr{a}$ if we assume that $I_\tnr{a}<I_\tnr{b}$.

Assume  $I_\tnr{a}=1, I_\tnr{b}=1.5$, and $p=0.75$ so $\ov{C}=1.375$. For  $K=2, 3$ we easily calculate the throughput\footnote{For $\R\leq 1$ we obtain $f_1=0$. For $1<\R\leq 1.5$ -- $f_1=1-p$ and $f_2=0$. For $1.5<\R\leq2$ -- $f_1=1$, $f_2=0$, etc.} as
\begin{align}\label{eta.IR}
\eta^{\IR}_2&=
\begin{cases}
\R, &\text{if}\quad \R\leq1\\
0.8\R, &\text{if}\quad 1 <\R \leq 1.5\\
0.5\R, &\text{if}\quad 1.5 <\R \leq 2
\end{cases},
\end{align}
and
\begin{align}
\eta^{\IR}_3&=
\begin{cases}
\eta^{\IR}_2, &\text{if}\quad \R\leq 2\\
0.48\R, &\text{if}\quad 2 <\R \leq 2.5\\
0.41\R, &\text{if}\quad 2.5 <\R \leq 3 
\end{cases}.
\end{align}

The optimum throughput-rate pairs are then $(\eta^{\IR}_2=1.2, \R=1.5)$ and $(\eta^{\IR}_3=1.23, \R=3)$. First, the benefit of using \gls{harq} is clear: we are able to transmit without errors with a finite number of channel blocks and go beyond the obvious limit of $I_\tnr{a}$. Second, we note that for $K=2$, after two transmissions, the accumulated \gls{mi} always satisfies $\Isig_2\geq 2$, while the condition $\Isig_2\geq 1.5$ is sufficient to decode the packet. This may be  seen as a ``waste'' which will be removed with the idea of cross-packet coding introduced in \secref{Sec:joint.codec}. 
\end{example}

\begin{example}[16QAM over Rayleigh fading channel]\label{Ex:QAM.Rayleigh}
Assume now that the transmission is done using symbols drawn uniformly from 16-points \gls{qam} constellation $\X$ \cite[Ch.~2.5]{Szczecinski_Alvarado_book} and that the channel gains follow Rayleigh distribution, \ie
\begin{align}\label{pdf.SNR}
\pdf_{\SNRrv}(\SNR)=1/\SNRav\exp(-\SNR/\SNRav),
\end{align}
where $\SNRav$ is the average \gls{snr}.

We calculate $I(\SNR)$ and the average $\ov{C}$ using the numerical methods outlined in \cite[Ch.~4.5]{Szczecinski_Alvarado_book} and compare it in \figref{Fig:Rayleigh.conventional} with the throughput $\eta^{\IR}_{\kmax}$ when $\kmax\in\set{2,\infty}$.\footnote{$\eta^{\IR}_{\infty}$ can be computed by taking $\kmax$ large enough in \eqref{eta.rr.detail} as suggested in  \cite{Larsson14} or by evaluating the throughput using the method outlined in the  \appref{Sec:OptimalAdaptation} and considering the policy $\pi(\mfs)=\R$ if $\mfs=(0,0)$ and $\pi(\mfs)=0$ otherwise. We opt for the later method.} The results indicate that i)~there is a significant loss with respect to the ergodic capacity when using truncated \gls{harq}, and ii)~increasing the number of transmission rounds ($K=\infty$) helps recovering the loss for a small-medium range of throughput (\eg for $\eta^{\IR}=1$ we gain $\sim 3\dBval$ and the gap to $\ov{C}$ is less than $1\dBval$), but it is less useful in the region of high $\eta^{\IR}_K$, \ie  in the vicinity of the maximum attainable throughput (\eg for $\eta^{\IR}=3$, we gain $1\dBval$  but the gap to $\ov{C}$ is still $\sim5\dBval$). We highlight this well-known effect \cite{Larsson14} to emphasize later the gains of the new coding strategy.
\end{example}

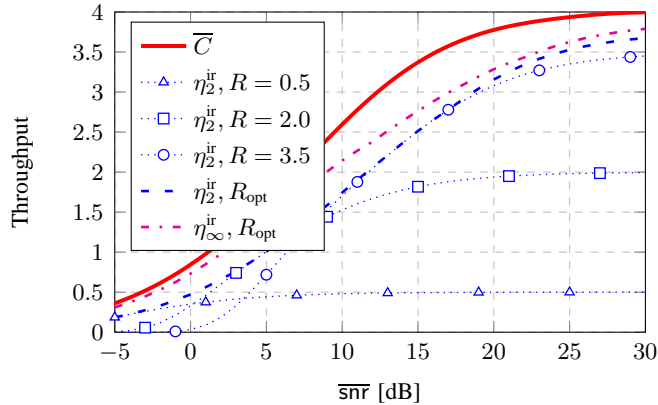
\begin{figure}
%
\centering
\newcommand{\lsize}{\footnotesize}
\pgfplotsset{tick label style={font=\footnotesize},
    label style={font=\footnotesize},
    legend style={font=\footnotesize}
}
\tikzset{every mark/.append style={mark size = 2, solid, fill=white, line width=0.5pt,}}
\begin{tikzpicture}
 \begin{axis}[
	axis lines = box,
	xlabel={\lsize $\SNRav$\dB},
	ylabel={\lsize  Throughput},
	xmin=-5,xmax=30,
	ymin=0,ymax=4,
	xtick={-5,0,...,30},
	ytick={0,0.5,..., 4},
	legend pos=north west,
	xmajorgrids=true,
	ymajorgrids=true,
	grid style=dashed,
	clip=true,
	legend style={
		draw=black,
		cells={anchor=west},
	},
]
\addplot[color=red, line width=1.5pt,]
table[x=SNRdB, y=Ergodic]{./figures/16QAM.Ergodic.Capacity.Rayleigh.dat};
\addlegendentry{$\ov{C}$}
\addplot[draw=blue, line width=0.5pt,dotted,mark=triangle*,
mark repeat=6,mark phase=1,]
table[x=SNRdB, y=etaR.0.5]{./figures/16QAM.C.HARQ.truncated.Rayleigh.dat};
\addlegendentry{ $\eta^{\IR}_2, \R=0.5$}
\addplot[draw=blue, line width=0.5pt,dotted,mark=square*,
mark repeat=6,mark phase=3,]
table[x=SNRdB, y=etaR.2.0]{./figures/16QAM.C.HARQ.truncated.Rayleigh.dat};
\addlegendentry{$\eta^{\IR}_2, \R=2.0$}
\addplot[draw=blue, mark color=white, line width=0.5pt,dotted,mark=*,mark color=blue,
mark repeat=6,mark phase=5,]
table[x=SNRdB, y=etaR.3.5]{./figures/16QAM.C.HARQ.truncated.Rayleigh.dat};
\addlegendentry{$\eta^{\IR}_2, \R=3.5$}
\addplot[color=blue, line width=1pt,loosely dashed,]
table[x=SNRdB, y=etaopt]{./figures/16QAM.C.HARQ.truncated.Rayleigh.dat};
\addlegendentry{$\eta^{\IR}_2, R_\tr{opt}$}
\addplot[color=magenta, line width=1pt,loosely dashdotted,]
table[x=SNRdB, y=etaopt]{./figures/16QAM.C.HARQ.persistent.Rayleigh.dat};
\addlegendentry{$\eta^{\IR}_\infty, R_\tr{opt}$}

\end{axis}

\end{tikzpicture}
\caption{Throughput of the conventional \gls{irharq}, compared to the ergodic capacity, $\ov{C}$, in Rayleigh block-fading channel. The $R_\tr{opt}$ curve is an envelope of the throughputs $\eta^{\IR}_{\kmax}$ obtained with different coding rates per block $\R\in\set{0.25,0.5,\ld, 7.75}$.}\label{Fig:Rayleigh.conventional}
\end{figure}

\section{Cross-packet HARQ}\label{Sec:joint.codec}

The examples  shown previously indicate that the conventional coding cannot bring the throughput of  \gls{harq} close to the capacity unless the nominal coding rate $\R$ and the number of rounds $K$ increase. We would like now to exploit a new coding possibility consisting in joint coding of packets during the \gls{harq} cycle.

Let us start with the case of two transmission rounds. In the first round, we use the nominal rate $\R_1$ is used, \ie the packet $\mfm_1\in\set{0,1}^{R_1\Ns}$ is encoded 
\begin{align}\label{encoding.1}
\bx_1=\Phi_1[\mfm_1] \in \X^{\Ns},
\end{align}
and transmitted over the channel \eqref{y.x.z} producing $\by_1=\sqrt{\SNR_1}\bx_1+\bz_1$, where $\Phi_k[\cd]$ is the encoding at the $k$th round. 

If the packet $\mfm_1$ is decoded correctly (which occurs if $I_1\ge\R_1$), a new cycle \gls{harq} starts by the transmission of a new packet. However, if the decoding fails, the packet  $\mfm_{[2]}=[\mfm_1,\mfm_2]\in\Binary^{(R_1+R_2) \Ns}$ is encoded using a conventional code designed independently of the codebook corresponding to the first transmission
\begin{align}\label{conventional.code}
\bx_2=\Phi_2[\mfm_1,\mfm_2] \in \mcX^{\Ns},
\end{align}
which yields the channel outcome $\by_2=\sqrt{\SNR_2}\bx_2+\bz_2$ as depicted in \figref{Fig:multi_message}.\footnote{This coding strategy is introduced without any claim of  optimality. The undeniable advantage of using independently generated codebooks is the simplicity of implementation. We note that the idea of using $\Phi_2$ independent of $\Phi_1$ was also proposed in \cite{Trillingsgaard14,Nguyen15}.}

\begin{figure}[tb]
\begin{center}
\resizebox{1\linewidth}{!}{\input{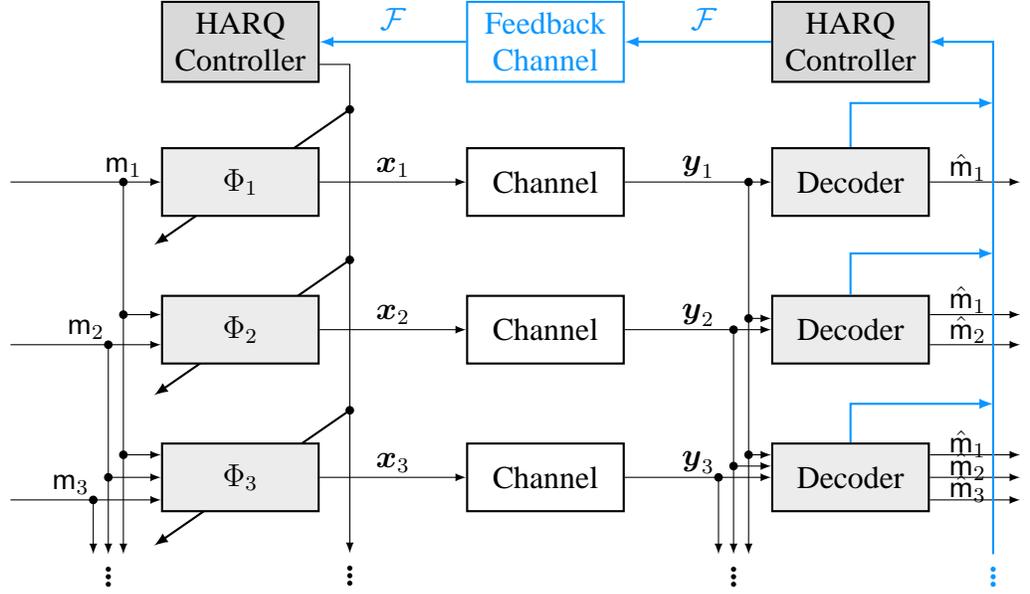}}
\end{center}
\caption{Model of the adaptive \gls{xp} transmission: the \gls{harq} controller uses the information $\mcF$ obtained over the feedback channel to choose the rate for the next round; $\mcF$ represent \gls{ack}/\gls{nack} acknowledgement in the case of one bit feedback, or, it carries the index of the coding rate in the case of rate-adaptive transmission (\secref{Sec:multibit}).}\label{Fig:multi_message}
\end{figure}

Intuitively, by introducing  $\mfm_2$ we want to prevent the ``waste'' of \gls{mi}, which happens if $\Isig_2$ is much larger than $\R_1$, \cf \exref{Ex:2states}.  After the second transmission, the receiver decodes the packets $[\mfm_1, \mfm_2]$ using the observations $\by_{[2]}=[\by_1, \by_2]$. The codebook obtained after two transmissions is illustrated in \figref{Fig:codebook}. 
The associated decoding conditions based on the channel outcomes $\by_{[2]}$ are given by
\begin{align}
\label{dec.1}
\Isig_2=I_1+I_2&\ge R_1+R_2,\\ 
\label{dec.2}
I_2&\ge R_2,
\end{align}
where \eqref{dec.1} is a constraint over the sum-rate that guarantees the joint decoding of the packets pair $(\mfm_1, \mfm_2)$ while \eqref{dec.2} ensures the correct decoding of the packet $\mfm_2$. This means, the \gls{mi} must be accumulated to decode each of the packets even though the decoding is done jointly. The formal proof of \eqref{dec.1} and \eqref{dec.2} is presented in the Appendix~\ref{Sec:Decoding_Proof}. Similar decoding conditions were presented in the context of \gls{phy} security in \cite{LeTreust13}. 

While the event $\nack_1$ remains unchanged with respect to the conventional coding, the event $\nack_2$ means that $\nack_1$ occurred, as well as, that \eqref{dec.1} and \eqref{dec.2} are not satisfied
\begin{align}\nonumber
\nack_2&=\Big\{\bigl(  I_1<R_1 \bigr) \wedge \ov{\bigl((\Isig_2\ge \Rsig_2) \wedge (I_2\ge R_2) \bigr)}\Big\}\\
\label{nack.2.a}
&=\Big\{ ( I_1<R_1) \wedge \bigl( (\Isig_2<\Rsig_2) \vee ( I_2<R_2) \bigr)  \Big\} \\
\label{nack.2}
&=\Big\{\bigl(  I_1<R_1 \bigr) \wedge \bigl(\Isig_2<\Rsig_2\bigr)\Big\},
\end{align}
where  $\Rsig_k\triangleq\sum_{l=1}^k \R_l$ and the event $\ov{E}$ is the complement of $E$. 
To pass from \eqref{nack.2.a} to \eqref{nack.2} we used the decoding failure implication
\begin{align}\nonumber
\set{I_1<R_1 \wedge I_2<R_2}\implies \set{  I_1<R_1  \wedge \Isig_2<\Rsig_2},
\end{align}
which means that $\nack_1$ combined with \eqref{dec.1} implies \eqref{dec.2}.

The above conditions  generalize straightforwardly for any $k>1$ with $\R_k$ being the rate of the packet $\mfm_k$ added in the $k$th round
\begin{align}
\label{nack.k}
\nack_k=\set{\nack_{k-1} \wedge \bigl(\Isig_k<\Rsig_k  \bigr)}.
\end{align}

To calculate the throughput of such an \gls{xp}, we adopt a similar approach as in \eqref{eta.rr.detail} but we must account for the reward in the $k$ transmission round given by $\Rsig_k$, which  yields
\begin{align}\nonumber
\eta^{\xp}_{\kmax}
&=\frac{\Rsig_1(1-f_1)+ \Rsig_2(f_1-f_2)+ \ld+\Rsig_K(f_{K-1}-f_K)}
{(1-f_1)+2\cd(f_1-f_2)+\ld+K\cd(f_{K-1})}\\
\label{throughput.joint}
&=\frac{\sum_{k=1}^{K}R_k \big(f_{k-1}-f_K\big) }{1+\sum_{k=1}^{K-1} f_k}.
\end{align}
Here, again $f_k=\PR{\nack_k}, k\ge1$ with $\nack_k$ defined by \eqref{nack.k}.

As a sanity check we can set $R_k=0, k=2, \ld, K$, and  recover the conventional single-packet \gls{harq}, \ie \eqref{throughput.joint} will be equivalent to \eqref{eta.rr}.

The fundamental difference of the proposed \gls{xp} with respect to the conventional \gls{harq} appears now clearly in the numerator of \eqref{throughput.joint} which expresses the idea of variable rate transmission due to encoding of multiple packets. Nevertheless, not only the numerator changed with respect to \eqref{eta.rr} but also the denominator is different due to the new definition of $\nack_k$ in \eqref{nack.k}. 

\begin{figure}[tb]
\psfrag{1}[l][l][0.7]{$1$}
\psfrag{Ns}[r][r][0.7]{$\Ns$}
\psfrag{Ns1}[l][l][0.7]{$\Ns+1$}
\psfrag{2Ns}[l][l][0.7]{$~2\Ns$}
\psfrag{2r1}[r][r][0.7]{$2^{R_1\Ns}$}
\psfrag{2r12}[l][l][0.7]{$2^{(R_1+R_2)\Ns}$}
\psfrag{vdots}[l][l][0.7]{$\vdots$}
\psfrag{PHI1}[l][l][0.7]{$\Phi_1$}
\psfrag{PHI2}[l][l][0.7]{$\Phi_2$}
\centering
\vskip 0.1in
\scalebox{1}{\includegraphics[width=0.5\linewidth]{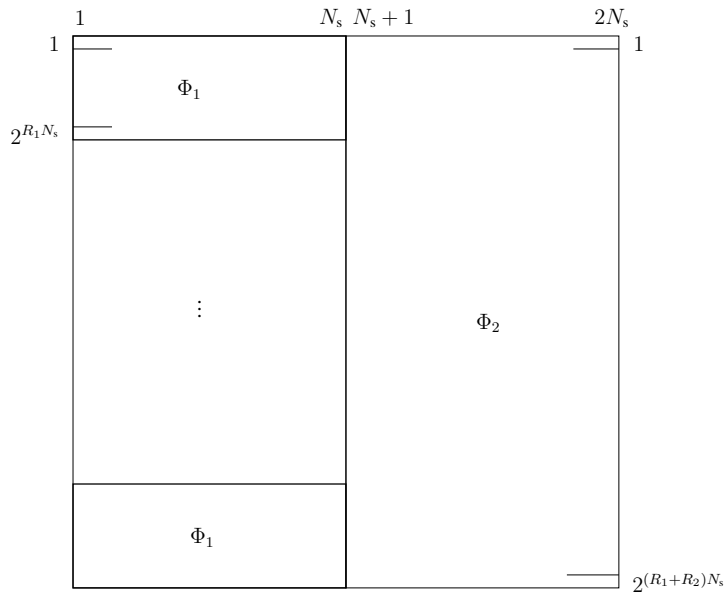}}
\caption{Illustration of the codebook defined through the coding function $\Phi_1$ in \eqref{encoding.1} and the joint coding function $\Phi_2$ in \eqref{conventional.code}. Each codeword composed of  $2\Ns$ symbols is indexed by the packet $\mf{m}_{[2]}$. The first $\Ns$ symbols are created without indexing by $\mf{m}_2$ so we artificially repeat them $2^{R_2\Ns}$ times to match the number of codewords in the codebook $\Phi_2$.}
\label{Fig:codebook}

\end{figure}

\begin{example}[Two-state channel and \gls{xp}]\label{Ex:2state.MP}
We consider now the proposed \gls{xp} in the scenario of \exref{Ex:2states}. Let us start, as before, with $\kmax=2$ and $R_1=1.5$. After a decoding failure (which means that we obtained $I_1=I_\tnr{a}=1$), we are free to define any rate $R_2$. In the absence of any formal criterion (more on that in \secref{Sec:Rate_Optimazation}), we take the following auxiliary (and somewhat ad-hoc) condition: we want to guarantee a non-zero successful decoding probability, \ie $f_2<1$. Here, since $\Isig_2\in(2,2.5)$, any  $R_2 \leq 1$ can ensure that $f_2<1$. In particular, if the rate  $R_2\leq 0.5$ we guarantee a much stronger condition $f_2=0$. 

For the case when $\kmax=2$ and using $R_2=0.5$, we obtain $f_1=0.25$ and $f_2=0$. The throughput is then given by
\begin{align}\label{eta.j.c.2}
\eta^{\xp}_2=\frac{\R_1+0.25\R_2  }{1+0.25}=1.3.
\end{align}

Thus, we used exactly the same channel resources as in the conventional \gls{harq}, obtained the same guarantee of successful decoding ($f_2=0$) after two transmission rounds, but the throughput is larger.

The difference is that, while we still have $\Isig_2\in(2,2.5)$, we now use $\Rsig_2=2$ to eliminated the ``waste'' of \gls{mi} in the conventional \gls{irharq}, where $\Rsig_2=1.5$. The improvement may be seen as the increase in the throughput (from $\eta^{\IR}_2=1.2$ to $\eta^{\xp}_2=1.3$) or as the reduction in the memory requirements (\ie we obtain a better throughput with smaller $K$, see $\eta^{\IR}_3=1.23$ in \exref{Ex:2states}). The price to pay for this advantages is the possible increase in complexity of cross-packet  encoding/decoding.

Similarly, for $\kmax=3$, we can use the larger value of $\R_2$ (that guarantees our objective of decodability, $f_2<1$), \ie $\R_2=1$. In this case, $f_1=0.25$, and $f_2=\PR{I_1<1.5 \wedge \Isig_2<2.5}=0.0625$. In the third transmission we observe $\Isig_{3}
\in(3,3.5)$ so, using $\R_3=0.5$, we obtain $f_3=0$ and thus the throughput is calculated as
\begin{align}\label{eta.j.c.3}
\eta_3^{\xp}=\frac{\R_1+0.25\R_2  + 0.0625\R_3}{1+0.25+0.0625}\approx 1.36,
\end{align}
which is already quite close to $\ov{C}=1.375$.
\end{example}

The improvement of the throughput in \gls{xp} is due to the way the codebook is constructed. While the conventional \gls{irharq}, see \secref{Sec:HARQ}, makes a rigid separation of the codewords into the fixed-content subcodewords -- an approach which is blind to the channel realizations, in \gls{xp} we match the information content of the codebook following the outcome of the transmissions.

\section{Optimization of the coding rates}\label{Sec:Rate_Optimazation}

Our goal now is to evaluate how well the \gls{xp} can perform. To this end, we will have to find the optimal coding rates $\R_1, \R_2, \ld, \R_K$ which maximize throughput \eqref{throughput.joint}. 

Since the objective function is highly non linear, we will use the exhaustive search: for a truncated \gls{harq} this can be done with a manageable complexity. 

\begin{example}[16QAM, Rayleigh fading -- continued]\label{Ex:allocation}
In \figref{Fig:Rayleigh.joint.truncated.allocation} we show the results of the exhaustive-search optimization  of ${\eta}^{\xp}_\kmax$ with ${\eta}^{\IR}_\kmax$; for implementability, we limited the search space: \gls{irharq} uses \mbox{$\R_1\in\set{0,0.25,\ld,3.75}$} and \gls{xp} uses rates which satisfy $\Rsig_{\kmax} \leq \R_{\max}$, with \mbox{$\R_{\max}=8$}; $\R_{1}\in\set{0.25,\ld,3.75}$, $\R_{k}\in\set{0,0.25,\ld,3.75} ~ \forall k\in\set{2,\ld,\kmax}$. 

We used here an additional constraints requires each transmission to have non zero probability of being decodable, that is $\R_{k} < \log_{2} M, \forall k=1,\ld,\kmax$, where $M=16$. In fact, these constraints were always satisfied in \gls{xp} so they only affect \gls{irharq}; we will relax them in the next example.

In terms of \gls{snr} required to attain $\eta=3$, the gain of \gls{xp} over \gls{irharq} varies from $1.5\dBval$ (for $\kmax =2$) to $2.5\dBval$ (for $K =3$).

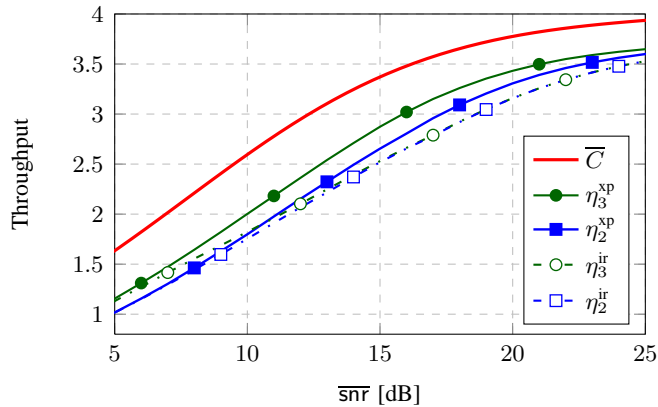
\begin{figure}[tb]
%

\centering
\newcommand{\lsize}{\footnotesize}
\pgfplotsset{tick label style={font=\footnotesize},
    label style={font=\footnotesize},
    legend style={font=\footnotesize}
}
\tikzset{every mark/.append style={mark size = 2.2, solid, fill=white, line width=0.5pt,}}
\begin{tikzpicture}
 \begin{axis}[
	axis lines = box,
	xlabel={\lsize $\SNRav$\dB},
	ylabel={\lsize  Throughput},
	xmin=5,xmax=25,
	ymin=0.8,ymax=4,
	xtick={5,10,...,30},
	ytick={0,0.5,..., 4},
	legend pos=south east,
	xmajorgrids=true,
	ymajorgrids=true,
	grid style=dashed,
	clip=true,
	legend style={
		draw=black,
		cells={anchor=west},
	},
]
\addplot[color=red, line width=1.2pt,]
table[x=SNRdB, y=Ergodic]{./figures/16QAM.Ergodic.Capacity.Rayleigh.dat};
\addlegendentry{$\ov{C}$}

\addplot[color=black!60!green, line width=0.8pt,solid,mark=*,mark repeat=5,mark phase=2,mark options={black!60!green}]
table[x=SNRdB, y=Throughput3]{./figures/16QAM.Joint.Coding.Truncated.Rayleigh.One.bit.dat};
\addlegendentry{$\eta^{\xp}_3$}

%
\addplot[color=blue,solid, line width=0.9pt,mark=square*,mark repeat=5,mark phase=4,mark options={blue}]
table[x=SNRdB, y=Throughput2]{./figures/16QAM.Joint.Coding.Truncated.Rayleigh.One.bit.dat};
\addlegendentry{$\eta^{\xp}_2$}
%





\addplot[color=black!60!green, line width=0.8pt,loosely dashdotted,mark=*,mark repeat=5,mark phase=3,]
table[x=SNRdB, y=Throughput3]{./figures/16QAM.C.HARQ.Truncated.Rayleigh.One.bit.dat};
\addlegendentry{$\eta^{\IR}_3$}

\addplot[color=blue, line width=0.8pt,loosely dashdotted,mark=square*,mark repeat=5,mark phase=5,]
table[x=SNRdB, y=Throughput2]{./figures/16QAM.C.HARQ.Truncated.Rayleigh.One.bit.dat};
\addlegendentry{$\eta^{\IR}_2$}

\end{axis}

\end{tikzpicture}
\caption{Throughput of the conventional \gls{irharq} ($\eta^{\IR}_\kmax$) compared to \gls{xp} ($\eta^{\xp}_\kmax$) in Rayleigh block-fading channel. The ergodic capacity ($\ov{C}$) is  shown for reference. }\label{Fig:Rayleigh.joint.truncated.allocation}
\end{figure}
\end{example}

\subsection{Rate adaptation}\label{Sec:multibit}
The possibility of varying the rates during the \gls{harq} cycle opens new optimization space and we want to explore it fully following the idea of adapting the transmission parameters in \gls{harq} on the basis of obsolete \gls{csi} considered before, \eg in \cite{Cheng03,LiuR03,Visotsky05,Pfletschinger10,Kim11,Szczecinski13,Jabi15}.

The idea is  to \emph{adapt}  the coding rates using obsolete \gls{csi}s, $I_1, I_2, \ld, I_{k-1}$; this concept remains compatible with the assumption of transmitter operating without \gls{csi} knowledge because the obsolete \gls{csi}s $I_1, I_2, \ld, I_{k-1}$ cannot be used in the $k$th round to infer anything about $I_{k}$ (due to \gls{iid} model of the \gls{snr}s). 

Using this approach, the rate $\R_k$ will not only depend on the \gls{mi}s $I_1,\ld, I_{k-1}$ but also -- on the past rates $\R_1, \ld, \R_k$.\footnote{Through $\R_1, \Rsig_2,\ld, \Rsig_{k-1}$, which determine the probability of the decoding success, see \eqref{nack.k}.} This recursive dependence may be dealt with using the  \gls{mdp} framework, where the states of the Markov chain not only indicate the transmission number but also gather all information necessary to decide on the rate, which in the language of the \gls{mdp} is called an \emph{action}. The state has to be defined so that i)~knowing the action (chosen rate), the state-transition probability can be determined after each transmission, and ii)~the reward may be calculated knowing the state and the action. The state defined as a pair $\mfs_k=(\Rsig_k,\Isig_k)$ satisfies these two requirements, where we only need to consider the pairs which satisfy $\Rsig_k>\Isig_k$, otherwise the decoding is successful and the \gls{harq} cycle terminates.

Thus, the rate adaptation consists in finding the functions (called \emph{policies}), $\R_{l}(\mfs_{l-1})$ maximizing the throughput, which is found generalizing the expression \eqref{throughput.joint}
\begin{align}\label{eq.mm.ad}
\hat{\eta}^{\xp}_{\kmax}&=\frac{\Ex\big[\sum_{k=1}^{\kmax}\xi_{k} \Rsig_k\big]}{1+\sum_{k=1}^{\kmax-1}f_{k}},
\end{align}
where 
\begin{align}\label{xi.original}
\xi_{k}=\IND{I_1<\R_1 \wedge \ld \wedge \Isig_{k-1}<\Rsig_{k-1}\wedge \Isig_{k}\geq \Rsig_k},
\end{align}
indicates the successful decoding in the $k$th round, and 
\begin{align}\label{Rsig.state}
\Rsig_k=\Rsig_{k-1}+\R_{k}(\mfs_{k-1})
\end{align}
is the accumulated rate depending in a recursive fashion on the states of the Markov chain. The probability of $k$ successive errors, $f_k$, may be expressed as \eqref{nack.k} considering the dependence of the rates on the states given by \eqref{Rsig.state}. All the expectations are taken with respect to the states -- or equivalently -- with respect to $I_1,\ld, I_{\kmax}$.

The expression \eqref{eq.mm.ad} will be useful in \secref{Sec:Heuristics}, however, its maximization with respect to the policies $\R_{l}(\mfs_{l-1}), l=1,\ld, \kmax$ will be done using efficient specialized algorithms as explained in \appref{Sec:OptimalAdaptation}. In the particular case of two \gls{harq} rounds ($\kmax=2$), the optimal rate adaptation policy can be derived in closed form as shown in  \appref{Sec:OptimalAdaptation.Analytic}.

To run the optimization algorithms outlined in \appref{Sec:OptimalAdaptation}, we need to discretize the variables involved (states and actions). As for the rates (actions), we use a relatively course discretization step equal to $0.25$ and define the action space as the set $\mcR=\set{0.25,0.5,\ld, \R_{\max}}$. While the results are notably affected by $\R_{\max}$, using a finer discretization step did not change the results significantly. 

Here,  it is natural to ask a question about the signaling overhead due to proposed adaptation scheme. We thus note that while we assume the outdated \gls{mi}, $\Isig_k$ is discretized with a  high resolution when optimizing the throughput (\cf \appref{Sec:OptimalAdaptation}), the feedback load is affected by the cardinality of the action space, $\mcR$: the receiver knows the accumulated \gls{mi} but only transmits the index of the chosen rate.

\begin{example}[16QAM, Rayleigh fading channel -- continued]\label{Ex:Rayleigh.joint.adaptation}
The throughput of adaptive \gls{xp}, $\hat{\eta}^{\xp}$, is compared to the throughput of the conventional \gls{irharq} in \figref{Fig:Rayleigh.joint.persistent} for $K=\infty$, while \figref{Fig:Rayleigh.joint.truncated} shows the comparison for truncated \gls{harq}. 

Here, for \gls{irharq}, we removed the constraints on the initial coding rate, $\R_1<\log_2 M$,  which were applied in \exref{Ex:allocation}. It allows us to increase the throughput $\eta^{\IR}_3$ at the cost of first transmission not being decodable. In our view this is  a potentially serious drawback but we show such results to complement those already shown in \figref{Fig:Rayleigh.joint.truncated.allocation}, where the decodability condition was imposed. Again, \gls{xp} was insensitive to the decodability constraints and always provided results with decodable transmissions.

\begin{figure}[tb] 
%
\centering
\newcommand{\lsize}{\footnotesize}
\pgfplotsset{tick label style={font=\footnotesize},
    label style={font=\footnotesize},
    legend style={font=\footnotesize}
}
\tikzset{every mark/.append style={mark size = 2, solid, fill=white, line width=0.5pt,}}
\begin{tikzpicture}
 \begin{axis}[
	axis lines = box,
	xlabel={\lsize $\SNRav$\dB},
	ylabel={\lsize  Throughput},
	xmin=5,xmax=30,
	ymin=1.3,ymax=4,
	xtick={5,10,...,30},
	ytick={0,0.5,..., 4},
	legend pos=south east,
	xmajorgrids=true,
	ymajorgrids=true,
	grid style=dashed,
	clip=true,
	legend style={
		draw=black,
		cells={anchor=west},
	},
]
\addplot[color=red, line width=1.2pt,]
table[x=SNRdB, y=Ergodic]{./figures/16QAM.Ergodic.Capacity.Rayleigh.dat};
\addlegendentry{$\ov{C}$}
\addplot[color=magenta,  line width=0.8pt, solid,mark=triangle*,mark repeat=7,mark phase=4, mark size=2.25, mark options={magenta}]
table[x=SNRdB, y=Throughput2]{./figures/16QAM.Joint.Coding.Persistent.Rayleigh.dat};
\addlegendentry{$\hat{\eta}^{\xp}_\infty,\R_{\max}=16$}
\addplot[color=cyan, line width=0.8pt, solid ,mark=diamond*,mark repeat=7,mark phase=3,mark size=2.25,mark options={cyan}]
table[x=SNRdB, y=Throughput1]{./figures/16QAM.Joint.Coding.Persistent.Rayleigh.dat};
\addlegendentry{$\hat{\eta}^{\xp}_\infty,\R_{\max}=8$}
\addplot[color=magenta,  line width=0.8pt,dashdotted, mark=triangle*,mark repeat=6,mark size=2.25,mark phase=4]
table[x=SNRdB, y=Throughpu2]{./figures/16QAM.C.HARQ.Persistent.Rayleigh.Different.R.max.dat};
\addlegendentry{$\eta^{\IR}_\infty,\R_{\max}=16$}
\addplot[color=cyan,  line width=0.8pt,dashdotted,mark=diamond*,mark repeat=7,mark size=2.25]
table[x=SNRdB, y=Throughpu1]{./figures/16QAM.C.HARQ.Persistent.Rayleigh.Different.R.max.dat};
\addlegendentry{$\eta^{\IR}_\infty,\R_{\max}=8$}

\end{axis}

\end{tikzpicture}
\caption{Optimal throughput of the conventional \gls{irharq} ($\eta^{\IR}_\infty$) compared to the proposed \gls{xp} ($\hat{\eta}^{\xp}_\infty$) in Rayleigh block-fading channel.  The ergodic capacity ($\ov{C}$) is  shown for reference. }\label{Fig:Rayleigh.joint.persistent}
\end{figure}
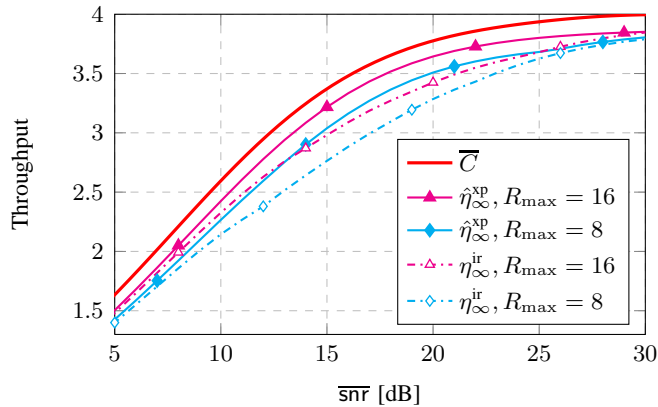

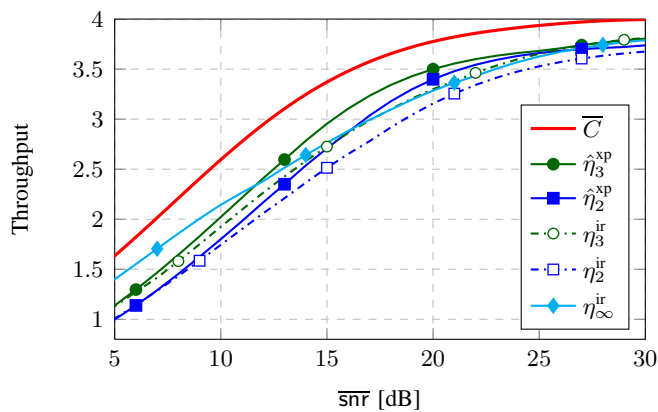
\begin{figure}[tb]
%

\centering
\newcommand{\lsize}{\footnotesize}
\pgfplotsset{tick label style={font=\footnotesize},
    label style={font=\footnotesize},
    legend style={font=\footnotesize}
}
\tikzset{every mark/.append style={mark size = 2, solid, fill=white, line width=0.5pt,}}
\begin{tikzpicture}
 \begin{axis}[
	axis lines = box,
	xlabel={\lsize $\SNRav$\dB},
	ylabel={\lsize  Throughput},
	xmin=5,xmax=30,
	ymin=0.8,ymax=4,
	xtick={5,10,...,30},
	ytick={0,0.5,..., 4},
	legend pos=south east,
	xmajorgrids=true,
	ymajorgrids=true,
	grid style=dashed,
	clip=true,
	legend style={
		draw=black,
		cells={anchor=west},
	},
]
\addplot[color=red, line width=1.2pt,]
table[x=SNRdB, y=Ergodic]{./figures/16QAM.Ergodic.Capacity.Rayleigh.dat};
\addlegendentry{$\ov{C}$}
\addplot[color=black!60!green, line width=0.8pt, solid,mark=*,mark repeat=7,mark phase=2,mark options={black!60!green}]
table[x=SNRdB, y=Throughput]{./figures/16QAM.Joint.Coding.K.3.Rayleigh.dat};
\addlegendentry{$\hat{\eta}^{\xp}_3$}
\addplot[color=blue, line width=0.8pt, solid,mark=square*,mark repeat=7,mark phase=2,mark options={blue}]
table[x=SNRdB, y=Throughput]{./figures/16QAM.Joint.Coding.K.2.Rayleigh.dat};
\addlegendentry{$\hat{\eta}^{\xp}_2$}
\addplot[color=black!60!green, line width=0.8pt,dashdotted,mark=*,mark repeat=7,mark phase=4,]
table[x=SNRdB, y=Throughput]{./figures/16QAM.C.HARQ.K.3.Rayleigh.dat};
\addlegendentry{$\eta^{\IR}_3$}
\addplot[color=blue, line width=0.8pt,dashdotted,mark=square*,mark repeat=6,mark phase=5,]
table[x=SNRdB, y=Throughput]{./figures/16QAM.C.HARQ.K.2.Rayleigh.dat};
\addlegendentry{$\eta^{\IR}_2$}
\addplot[color=cyan, line width=0.8pt,mark=diamond*,mark repeat=7,mark phase=3,mark size=2.5,mark options={cyan}]
table[x=SNRdB, y=Throughpu1]{./figures/16QAM.C.HARQ.Persistent.Rayleigh.Different.R.max.dat};
\addlegendentry{$\eta^{\IR}_\infty$}
\end{axis}

\end{tikzpicture}
\caption{Throughput of the conventional \gls{irharq} ($\eta^{\IR}_\kmax$) compared to the proposed \gls{xp} ($\hat{\eta}^{\xp}_\kmax$) for a truncated \gls{harq}, $\kmax\in\set{2,3}$ in Rayleigh block-fading channel;  $\R_{\max}=8$. The ergodic capacity ($\ov{C}$) and the optimal throughput of the persistent conventional \gls{irharq} ($\eta^{\IR}_\infty$) are  shown for reference.}\label{Fig:Rayleigh.joint.truncated}
\end{figure}
The improvements due to adaptive \gls{xp} are most notable for high values of the throughput. In particular we observe that
\begin{itemize}
\item The persistent \gls{xp} halves the gap between the ergodic capacity and the conventional \gls{irharq}. For example, the \gls{snr} gap  between $\hat{\eta}^{\xp}_\infty=3$ and the ergodic capacity, $\ov{C}=3$ is reduced by more than $50\%$ when comparing to the gap between $\eta^{\IR}_\infty=3$ and $\ov{C}=3$ which is equal to $5\dBval$ when $\R_{\max}=8$. We note that the throughput of \gls{xp} increases when $\R_{\max}$ increases: the \gls{snr} gap between $\ov{C}$ and $\hat{\eta}^{\xp}_\infty$ is reduced by half when $\R_{\max}=16$ is used instead of $\R_{\max}=8$.
\item For any value of throughput $\eta>3$, two rounds of \gls{xp} yield higher throughput than the conventional persistent \gls{irharq}. Thus, in this operation range we may improve the performance and yet decrease the memory requirements at the receiver.
\end{itemize}

\end{example}

\subsection{Heuristic adaptation policy}\label{Sec:Heuristics}

\begin{figure}[tb] 
%
\centering
\newcommand{\lsize}{\footnotesize}
\pgfplotsset{tick label style={font=\footnotesize},
    label style={font=\footnotesize},
    legend style={font=\footnotesize}
}
\tikzset{every mark/.append style={mark size = 2, solid, fill=white, line width=0.5pt,}}
\begin{tikzpicture}
 \begin{axis}[
	axis lines = box,
	xlabel={\lsize $\Rsig_{k-1}-\Isig_{k-1}$},
	ylabel={\lsize  $\R_{k}$},
	xmin=0.1270,xmax=4,
	ymin=0,ymax=4,
	xtick={0,0.5,...,4},
	ytick={0,0.5,..., 8},
	legend pos=north east,
	xmajorgrids=true,
	ymajorgrids=true,
	grid style=dashed,
	clip=true,
	legend style={
		draw=black,
		cells={anchor=west},
	},
]
\addplot[color=red, line width=1.2pt,]
table[x=MI1, y=policy1]{./figures/16QAM.Joint.Coding.Persistent.Rayleigh.Policy.dat};
\addlegendentry{$\Rsig_{k-1}=2.5$}
\addplot[color=blue, line width=1.2pt,dashdotted]
table[x=MI2, y=policy2]{./figures/16QAM.Joint.Coding.Persistent.Rayleigh.Policy.dat};
\addlegendentry{$\Rsig_{k-1}=3.5$}
\addplot[color=green, line width=1.2pt,dashed]
table[x=MI3, y=policy3]{./figures/16QAM.Joint.Coding.Persistent.Rayleigh.Policy.dat};
\addlegendentry{$\Rsig_{k-1}=5$}
\addplot[color=magenta, line width=1.2pt,dotted]
table[x=MI4, y=policy4]{./figures/16QAM.Joint.Coding.Persistent.Rayleigh.Policy.dat};
\addlegendentry{$\Rsig_{k-1}=6$}


\end{axis}

\end{tikzpicture}
\caption{Optimal rate $R_k$ as a function of $\Rsig_{k-1}-\Isig_{k-1}$ for different values of $\Rsig_{k-1}$; $\kmax=\infty$, $\SNRav=20\dBval$, $\R_{\max}=8$.}\label{Fig:Strategy.persistent}
\end{figure}
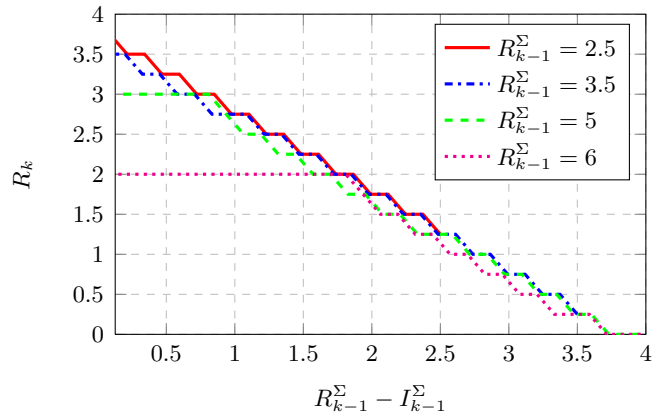

\figref{Fig:Strategy.persistent} shows the optimal rate adaptation as a function of $\Rsig_{k-1}-\Isig_{k-1}$ for different values of $\Rsig_{k-1}$, where we note a quasi-linear behaviour of the adaptation function with the saturation which occurs to guarantee $\Rsig_{k-1}+\R_k\leq\R_{\max}$. 

To exploit this very regular form, which was also observed solving the related problems  in \cite{Szczecinski13,Jabi15}, we propose to use the following heuristic function inspired by \figref{Fig:Strategy.persistent}
\begin{align}\label{Rk.heuristic}
\R_k=\R_1- (\Rsig_{k-1}-\Isig_{k-1}),
\end{align}
where only the rate $\R_1$ needs to be optimized (from \figref{Fig:Strategy.persistent} we find $\R_1\approx3.5$). Furthermore, applying \eqref{Rk.heuristic} recursively we obtain  $\R_2=I_1, \R_3=I_2, \ld, \R_k=I_{k-1}$; the identical rate-adaptation strategy may be derived from \cite[Sec.~III]{Popovski14}.

The simplicity of the adaptation function allows us now to evaluate analytically the throughput of \gls{xp}. To this end we need to calculate $f_l$ in the denominator of \eqref{eq.mm.ad} and the expectation in its numerator.

We first note that, from \eqref{Rk.heuristic} we obtain 
\begin{align}\label{}
\big(\Isig_k&<\Rsig_k\big) \iff  (I_k<R_1),
\end{align}
which means that the probability of decoding failure does not change with the index of the transmission round. Thus
\begin{align}\label{eq.fk.heur}
f_{k}&=(f_{1})^k,
\end{align}
and \eqref{xi.original} may be formulated as
\begin{align}\label{eq.xi.heur}
\xi_{k}=\Big(\prod_{l=1}^{k-1}\IND{I_l<\R_1}\Big)\IND{I_{k}\geq \R_1}.
\end{align}

From \eqref{Rk.heuristic} we also obtain $\Rsig_k=\R_1 + \sum_{l=1}^{k-1}I_l$, which allows us to calculate the expectation in the numerator of \eqref{eq.mm.ad} as
\begin{align}\label{}
\Ex[\xi_{k} \Rsig_k]
&=\Ex[\xi_{k} (\R_1 + I_1+\ld, I_{k-1})]\\
\label{xi.final}
&=\big(\R_1f_1 +(k-1)\tilde{C}\big)(f_1)^{k-2}(1-f_1),
\end{align}
where $\tilde{C}=\Ex_{I_1}\big[I_{1}\cd\IND{I_{1}<\R_{1}}\big]$ is a ``truncated'' expected \gls{mi}.

Using \eqref{xi.final} and \eqref{eq.fk.heur} in \eqref{eq.mm.ad}, the throughput is calculated as 
\begin{align}\label{eq.xi.heur.final}
\tilde{\eta}^\xp_\kmax&=\R_{1}(1-f_{1})+\frac{\tilde{C}(1-f_{1})}{1-f_{1}^\kmax} \nonumber\\
&\qquad\times \Big(-(\kmax-1)f_{1}^{\kmax-1}+\frac{1-f_{1}^{\kmax-1}}{1-f_{1}} \Big).
\end{align}

In the limit, $\kmax\rightarrow\infty$, \eqref{eq.xi.heur.final} becomes
\begin{align}\label{eq.xi.heur.infini}
\tilde{\eta}^{\xp}_\infty=\R_{1}(1-f_{1})+\tilde{C},
\end{align}
which is the same as \cite[Eq.~(12)]{Popovski14}.

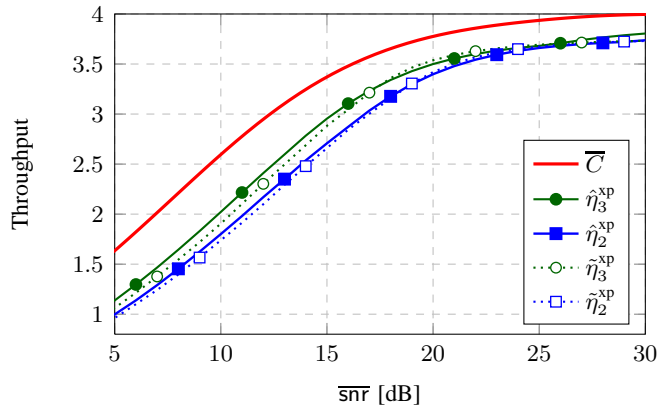
\begin{figure}
%

\centering
\newcommand{\lsize}{\footnotesize}
\pgfplotsset{tick label style={font=\footnotesize},
    label style={font=\footnotesize},
    legend style={font=\footnotesize}
}
\tikzset{every mark/.append style={mark size = 2, solid, fill=white, line width=0.5pt,}}
\begin{tikzpicture}
 \begin{axis}[
	axis lines = box,
	xlabel={\lsize $\SNRav$\dB},
	ylabel={\lsize  Throughput},
	xmin=5,xmax=30,
	ymin=0.8,ymax=4,
	xtick={5,10,...,30},
	ytick={0,0.5,..., 4},
	legend pos=south east,
	xmajorgrids=true,
	ymajorgrids=true,
	grid style=dashed,
	clip=true,
	legend style={
		draw=black,
		cells={anchor=west},
	},
]
\addplot[color=red, line width=1.2pt,]
table[x=SNRdB, y=Ergodic]{./figures/16QAM.Ergodic.Capacity.Rayleigh.dat};
\addlegendentry{$\ov{C}$}
\addplot[color=black!60!green, line width=0.8pt,solid,mark=*,mark repeat=5,mark phase=2,mark options={black!60!green}]
table[x=SNRdB, y=Throughput]{./figures/16QAM.Joint.Coding.K.3.Rayleigh.dat};
\addlegendentry{$\hat{\eta}^\xp_3$}

\addplot[color=blue,solid, line width=0.9pt,mark=square*,mark repeat=5,mark phase=4,mark options={blue}]
table[x=SNRdB, y=Throughput]{./figures/16QAM.Joint.Coding.K.2.Rayleigh.dat};
\addlegendentry{$\hat{\eta}^\xp_2$}

\addplot[color=black!60!green, line width=0.8pt,dotted,mark=*,mark repeat=5,mark phase=3,]
table[x=SNRdB, y=Throughput3]{./figures/16QAM.Petar.truncated.Rayleigh.Multi.Bit.dat};
\addlegendentry{$\tilde{\eta}_3^\xp$}

\addplot[color=blue, line width=0.8pt,dotted,mark=square*,mark repeat=5,mark phase=5,]
table[x=SNRdB, y=Throughput2]{./figures/16QAM.Petar.truncated.Rayleigh.Multi.Bit.dat};
\addlegendentry{$\tilde{\eta}_2^\xp$}

\end{axis}

\end{tikzpicture}
\caption{Throughput  of the optimal \gls{xp} ($\hat{\eta}^\xp_\kmax$) is compared to the throughput of \gls{xp} with the heuristic policy ($\tilde{\eta}^\xp_\kmax$) in Rayleigh block-fading channel. The ergodic capacity ($\ov{C}$) is shown for reference.}\label{Fig:Rayleigh.joint.vs.Petar.truncated.adaptation}
\end{figure}

\begin{example}[16QAM, Rayleigh fading -- continued]\label{Ex:Rayleigh.joint.vs.BRQ.adaptation.and.allocation}
We compare in \figref{Fig:Rayleigh.joint.vs.Petar.truncated.adaptation} the throughput of optimal  \gls{xp} with the heuristic policy \eqref{Rk.heuristic}, which is optimized over $\R_{1}$. As expected, the optimal solution outperforms the heuristic policy but the gap is very small (less than $0.5\dBval$). Moreover, since  $\hat{\eta}^\xp_\kmax$ was optimized over a finite set of rates $\mcR=\set{0.25,0.5,\ld,\R_{\max}}$, and the heuristic policy assumes that $\mcR$ is continuous and unbounded, $\tilde{\eta}^\xp_\kmax$ slightly outperforms $\hat{\eta}^\xp_\kmax$ above $\SNRav=20\dBval$. This gap can be reduced increasing the value of $\R_{\max}$; decreasing the discretisation step below $0.25$ had much lesser influence on the results.\end{example}

The results are quite intriguing and suggesting  that the strategy of \cite{Popovski14} based on a double-layer encoding\footnote{\cite{Popovski14} proposes double-step encoding: to form $\mfm_{[k]}$ the bits $\mfm_{k}$ and the parity bits of $\mfm_{[k-1]}$ are first ``mixed'', and next, the channel encoder is used.} and a transmission-by-transmission decoding (as opposed to the joint decoding required in \gls{xp}), asymptotically yield the same throughput as the heuristic cross-packet \gls{harq}, whose throughput is also very close to the optimal \gls{xp}.

We cannot follow that path here but this  relationship should be studied in more details; in particular, the effect of removing the idealized assumption of using a continuous set of rates $\mcR$, necessary to implement \eqref{Rk.heuristic}, should be analyzed.

\section{Example of a practical implementation}\label{Sec:Practical}

Until now, we have adopted the perfect decoding assumption, \ie the decoding error in the $k$th round is equivalent to the event $\set{I_1<\R_1\wedge\ld\wedge \Isig_{k}<\Rsig_{k}}$. We will remove now this idealization to highlight also the practical aspect of \gls{xp}. 

We thus implement the cross-packet encoders in \figref{Fig:multi_message} using turbo encoders. To this end, as shown in \figref{Fig:HARQ.Turbo} we separate each encoder $\Phi_k$ into i)~a bit-level multiplexer, $\mcM$, whose role is to interleave the input packets $\mfm_1, \ld, \mfm_k$ and produce the packet, $\mfm_{[k]}$,  ii)~a conventional turbo-encoder (TC), iii)~the rate-matching puncturer, $\mcP$, which ensures that all binary codewords $\bc_k$ have the same length, $\Nc$, and iv)~a modulator, which maps the codewords $\bc_k$ onto the codewords $\bx_k$ from the constellation $\mcX$;  since we use $16$ary \gls{qam}, $\Nc=\Ns \log_2(M)$.

The multiplexers $\mcM_k$ are implemented using pseudo-random interleaving. The encoders (TC) are constructed via parallel concatenation of two recursive convolutional encoders with polynomials $[13/15]_8$. Each TC produces a $N_{\tr{b},[k]}=\Ns\R_1+\ld+\Ns\R_k$ systematic (input) bits and $N_{\tr{p}}=2N_{\tr{b},[k]}$ parity bits $\bp_k$.\footnote{We neglects the effect of the trellis terminating bits.} The bits $\bc_k$ are obtained concatenating ``fresh'' systematic bits $\mfm_k$ (those which  were not transmitted in the previous rounds) and the parity bits selected from $\bp_k$ via a periodic puncturing.

Such a construction of the encoders is of course not optimal and better interleavers and puncturers may be sought; however, their optimal design represents a challenge of its own and must be considered out of scope of the example we present here.

The encoding is rather straightforward and can be implemented using conventional elements. The decoding in the $k$th round is slightly more involved because it is done using outcomes of all transmissions, $\by_{[k]}$. From this perspective, we may see the binary codewords $\bc_1, \ld, \bc_k$ as an outcome of $2k$ concatenated convolutional encoders (two encoders per \gls{harq} round), each producing the sequence with increasing lengths. The decoding of multiple encoding units was already addressed before  \cite{huettinger2002design}\cite{divsalar1995multiple} and requires implementation of $2k$ \gls{bcjr} decoders (one for each of the encoders) exchanging the extrinsic probabilities for the information bits. We implement the serial scheduling, that is, once a \gls{bcjr} decoder is activated, it must wait  till all other \gls{bcjr} decoders are activated. One iteration is defined as $2k$ activations. The results we present are obtained using algorithm from the library \cite{Doray15}; we use $\Ns=1024$ and four decoding iterations.

\usetikzlibrary{arrows}
\usetikzlibrary{shapes}
\usetikzlibrary{positioning}
\usetikzlibrary{calc}
\usetikzlibrary{shapes.misc}
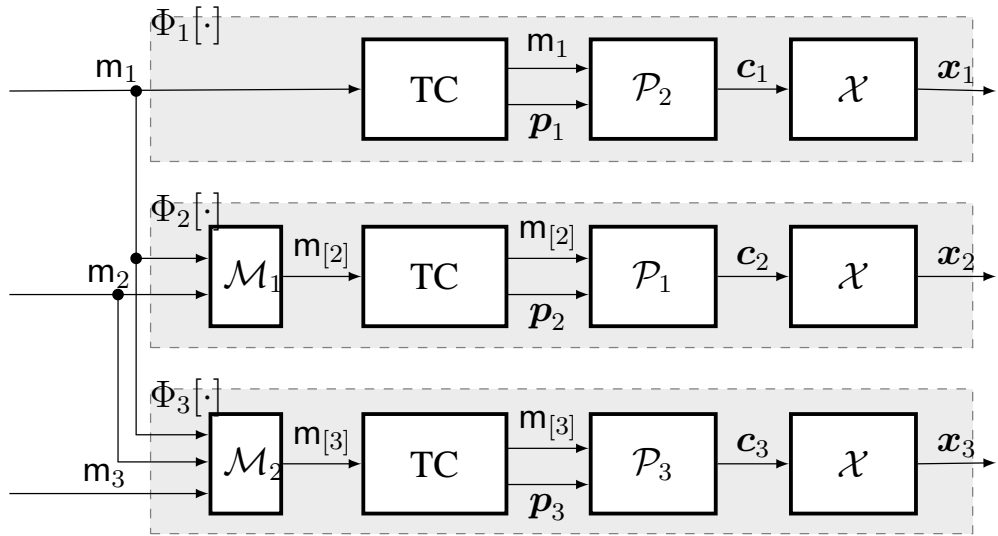
\begin{figure}[bt]
\begin{center}
\resizebox{1\linewidth}{!}{\pgfdeclarelayer{background}
\pgfdeclarelayer{foreground}
\pgfsetlayers{background,main,foreground}

\definecolor{LHCblue}{RGB}{10, 150, 255}

\tikzstyle{form0} = [draw, minimum width=0cm, text width=0cm, 
  text centered,  minimum height=0cm]

\tikzstyle{form1} = [draw,  thick, minimum width=1cm, text width=1.3cm,
  text centered,  minimum height=1.1cm ,very thick, fill=white]
  
  \tikzstyle{form01} = [draw , color=gray,minimum width=1cm, text width=8.8cm,
  text centered,  minimum height=1.6cm , dashed , fill=gray!15]
  
\tikzstyle{form11} = [draw, minimum width=0.9cm, text width=1.1cm,
  text centered,  minimum height=1.1cm ,very thick,  fill=white]

\tikzstyle{form2} = [draw=none, minimum width=3.4cm, text width=1.5cm, fill=none, 
  text centered,  minimum height=0cm]    

\tikzstyle{form3} = [color=LHCblue,draw, thick,minimum width=1.6cm, text width=1.5cm, 
  text centered,  minimum height=1.4cm]
  
  \tikzstyle{form4} = [draw=none, minimum width=0cm, text width=0cm, fill=none, 
  text centered,  minimum height=.6cm] 
  
    \tikzstyle{form5} = [draw, minimum width=.4cm, text width=.5cm,
text centered,  minimum height=1.1cm, very thick, fill=white]

 \tikzstyle{form05} = [draw=none, minimum width=.4cm, text width=.3cm,
text centered,  minimum height=0.5cm, very thick , color=black, fill=none]

    \tikzstyle{form7} = [draw, minimum width=0.8cm, text width=0.8cm,fill=blue!20, 
text centered,  minimum height=0.6cm]

     \tikzstyle{form6} = [draw, minimum width=2cm, text width=.4cm, 
  text centered,  minimum height=1.5cm]
  
      \tikzstyle{form8} = [draw, minimum width=1.8cm, text width=1.8cm, fill=green!50!yellow!15,
  text centered,  minimum height=2.5cm]
  
   \tikzstyle{form9} = [draw, minimum width=2.5cm, text width=3cm, fill=green!50!yellow!15,
  text centered,  minimum height=4cm]
  
    \tikzstyle{form10} = [draw, minimum width=0.8cm, text width=0.8cm,fill=yellow!40, 
text centered,  minimum height=0.6cm]

\begin{tikzpicture}

    \node(enc2)[form1] { TC};    

    \path (enc2.north)+(0,1.5)  node[form1](enc1) {TC}; 
     
     \path (enc2.south)+(0,-1.5)  node[form1](enc3) {TC}; 
     
     \path(enc2.east)+(-2.9,0) node[form5](inter1){$\mcM_1$};    
     
      \path(enc3.east)+(-2.9,0) node[form5](inter2){$\mcM_2$};    
      
       \path (enc2.east)+(1.6,0)  node[form11](punc2) {$\mcP_1$}; 
     
       \path (enc1.east)+(1.6,0)  node[form11](punc1) {$\mcP_2$}; 
     
       \path (enc3.east)+(1.6,0)  node[form11](punc3) {$\mcP_3$}; 

        \path (punc2.east)+(1.5,0)  node[form11](mod2) {$\mcX$}; 
     
       \path (punc1.east)+(1.5,0)  node[form11](mod1) {$\mcX$}; 
     
       \path (punc3.east)+(1.5,0)  node[form11](mod3) {$\mcX$}; 
       
        \begin{pgfonlayer}{background} 
      \path (enc2.north)+(1.4,1.5)  node[form01](cont1){}; 
       \path (cont1.south)+(0,-1.25)  node[form01](cont2) {}; 
      \path (cont2.south)+(0,-1.25)  node[form01](cont3) {}; 
          \end{pgfonlayer}
          
       \path (cont1.west)+(0.15,0.7)  node[form05](phi1) {$\Phi_1[\cd]$}; 
        \path (cont2.west)+(0.15,0.7)  node[form05](phi2) {$\Phi_2[\cd]$}; 
          \path (cont3.west)+(0.15,0.7)  node[form05](phi3) {$\Phi_3[\cd]$};

    \draw [draw,->,>=latex] (-4.7,2.05)--(enc1.west)node[shift={(-2.7,0.2)}]{$\mfm_1$};         
    
        \path[draw,->,>=latex] (-4.7,-0.2)--node[shift={(0.,.2)}]{$\mfm_2$}(-2.5,-0.2);    
           
       \path[draw,->,>=latex] (-4.7,-2.4)--node[shift={(-0.05,.2)}]{$\mfm_3$}(-2.5,-2.4);

          \path[draw,->,>=latex] (mod1.east)--node[shift={(0.,.2)}]{$\bx_1$}(6.2,2.05);    
      
       \path[draw,->,>=latex] (mod2.east)--node[shift={(0.,.2)}]{$\bx_2$}(6.2,0);   
           
       \path[draw,->,>=latex] (mod3.east)--node[shift={(0.,.2)}]{$\bx_3$}(6.2,-2.06);

             \path[draw,->,>=latex] (punc1.east)--node[shift={(0.,.2)}]{$\bc_1$}(mod1.west);    
      
       \path[draw,->,>=latex] (punc2.east)--node[shift={(0.,.2)}]{$\bc_2$}(mod2.west);   
           
       \path[draw,->,>=latex] (punc3.east)--node[shift={(0.,.2)}]{$\bc_3$}(mod3.west);

            \path[draw,->,>=latex] (0.8,2.3)--node[shift={(0.,.25)}]{$\mfm_{1}$}(1.7,2.3); 
            \path[draw,->,>=latex] (0.8,1.9)--node[shift={(0.,-.25)}]{$\bp_1$}(1.7,1.9);

         \path[draw,->,>=latex] (0.8,0.2)--node[shift={(0.,.25)}]{$\mfm_{[2]}$}(1.7,0.2); 
            \path[draw,->,>=latex] (0.8,-0.2)--node[shift={(0.,-.25)}]{$\bp_2$}(1.7,-0.2); 
       
        \path[draw,->,>=latex] (0.8,-1.9)--node[shift={(0.,.25)}]{$\mfm_{[3]}$}(1.7,-1.9); 
            \path[draw,->,>=latex] (0.8,-2.3)--node[shift={(0.,-.25)}]{$\bp_3$}(1.7,-2.3);
       
         \path[draw,->,>=latex] (inter1.east)--node[shift={(0.,.25)}]{$\mfm_{[2]}$}(enc2.west);   
         \path[draw,->,>=latex] (inter2.east)--node[shift={(0.,.25)}]{$\mfm_{[3]}$}(enc3.west);
         
        \draw [draw,-,>=latex] (-3.3,2.05)--node[shift={(-0.2,-0.2)}]{}(-3.3,-1.75);  
          \draw [draw,->,>=latex] (-3.3,0.2)--node[shift={(-0.2,-.2)}]{}(-2.5,0.2);  
         \draw [draw,->,>=latex] (-3.305,-1.75)--node[shift={(-0.2,-.2)}]{}(-2.5,-1.75);  

          \draw [draw,-,>=latex] (-3.5,-0.2)--node[shift={(-0.2,-0.2)}]{}(-3.5,-2.05);  
         \draw [draw,->,>=latex] (-3.505,-2.05)--node[shift={(-0.2,-.2)}]{}(-2.5,-2.05);  


    \draw [thick,fill=black] (-3.3,2.05)circle (.5mm); 
    \draw [thick,fill=black] (-3.3,0.2)circle (.5mm); 
     \draw [thick,fill=black] (-3.5,-0.2)circle (.5mm); 

\end{tikzpicture}}
\caption{Implementation of the encoders $\Phi_k[\cd]$ using turbo codes (TC), bit multiplexing ($\mcM_k$), puncturing ($\mcP$), and modulation ($\mcX$).}\label{Fig:HARQ.Turbo}
\end{center}
\end{figure}

Since we do not have the closed-form formula which describes the probability of error under particular channel conditions, especially when multiples transmissions are involved, the rate-adaptation approach seems to be out of reach and we focus on finding the fixed coding rates $\R_k, k=1,\ld,\kmax$. We use the brute search over the space of available coding rates which verifies the following conditions $\sum_{k=1}^{\kmax}\R_{k}\leq 8$, $\R_{1}\in\set{1.5,1.75,2,\ld,3.75}$, $\R_{k}\in\set{0,0.25,\ld,3.75}, \forall k>1$. 

The results obtained are shown in \figref{Fig:Rayleigh.joint.truncated.allocation.turbo} where the \gls{snr} gap (for the throughput $\eta=3$) between \gls{xp} and the conventional \gls{irharq} is $\sim1.5\dBval$ for $\kmax=2$ and $\sim2\dBval$ for $\kmax=3\dBval$. We attribute a small improvement of the throughput $\eta^{\xp}_3$ over $\eta^{\xp}_2$ to the suboptimal encoding scheme we consider in this example. 

We also note that the improvement of $\eta^{\IR}_3$ with respect to $\eta^{\IR}_2$  does not materialize. This is because \gls{irharq} is optimized for $\R_1$ but, due to limitation of the turbo encoder which generates only $3\Nb$ bits, a full redundancy cannot be always obtained and,  in such a case, we are forced to repeat the systematic and parity bits.  This explains why $\eta^{\IR}_3$ and $\eta^{\IR}_2$ are very similar for low throughput. On the other hand, they should be, indeed, similar for high throughput as we have seen in the numerical examples before.  

We show in \figref{Fig:Rayleigh.joint.truncated.allocation.turbo} the ergodic capacity where the gap to the throughput of the TC-based transmission is increased by additional ~$3\dBval$ which should be expected when using relatively-short codewords and practical decoders.

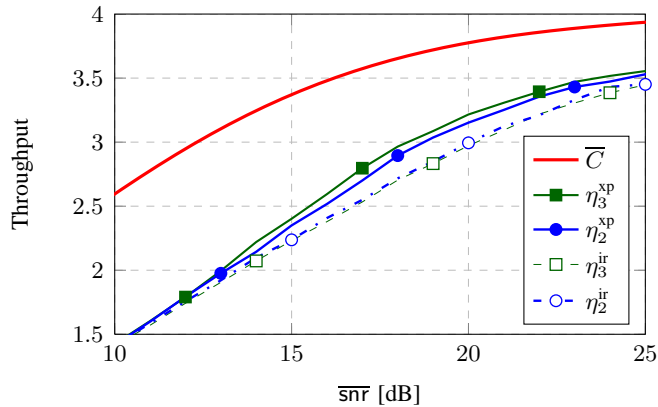
\begin{figure}[tb]
%

\centering
\newcommand{\lsize}{\footnotesize}
\pgfplotsset{tick label style={font=\footnotesize},
    label style={font=\footnotesize},
    legend style={font=\footnotesize}
}
\tikzset{every mark/.append style={mark size = 2.2, solid, fill=white, line width=0.5pt,}}
\begin{tikzpicture}
 \begin{axis}[
	axis lines = box,
	xlabel={\lsize $\SNRav$\dB},
	ylabel={\lsize  Throughput},
	xmin=10,xmax=25,
	ymin=1.5,ymax=4,
	xtick={5,10,...,30},
	ytick={0,0.5,..., 4},
	legend pos=south east,
	xmajorgrids=true,
	ymajorgrids=true,
	grid style=dashed,
	clip=true,
	legend style={
		draw=black,
		cells={anchor=west},
	},
]
\addplot[color=red, line width=1.2pt,]
table[x=SNRdB, y=Ergodic]{./figures/16QAM.Ergodic.Capacity.Rayleigh.dat};
\addlegendentry{$\ov{C}$}

\addplot[color=black!60!green, line width=0.8pt,solid,mark=square*,mark repeat=5,mark phase=2,mark options={black!60!green}]
table[x=SNRdB, y=Throughput3]{./figures/16QAM.Joint.Coding.Truncated.Rayleigh.One.bit.Turbo.code.dat};
\addlegendentry{$\eta^{\xp}_3$}

\addplot[color=blue,solid, line width=0.9pt,mark=*,mark repeat=5,mark phase=3,mark options={blue}]
table[x=SNRdB, y=Throughput2]{./figures/16QAM.Joint.Coding.Truncated.Rayleigh.One.bit.Turbo.code.dat};
\addlegendentry{$\eta^{\xp}_2$}



\addplot[color=black!60!green, line width=0.pt,dashdotted,mark=square*,mark repeat=5,mark phase=4,]
table[x=SNRdB, y=Throughput3]{./figures/16QAM.C.HARQ.Truncated.Rayleigh.One.bit.Turbo.code.dat};
\addlegendentry{$\eta^{\IR}_3$}

\addplot[color=blue, line width=0.9pt,loosely dashdotted,mark=*,mark repeat=5,mark phase=5,]
table[x=SNRdB, y=Throughput2]{./figures/16QAM.C.HARQ.Truncated.Rayleigh.One.bit.Turbo.code.dat};
\addlegendentry{$\eta^{\IR}_2$}

\end{axis}

\end{tikzpicture}
\caption{Turbo-coded transmission: the conventional \gls{irharq} ($\eta_\kmax$) is compared to \gls{xp} ($\eta^{\xp}_\kmax$)  in Rayleigh block-fading channel. }\label{Fig:Rayleigh.joint.truncated.allocation.turbo}
\end{figure}
%


\section{Conclusions}\label{Sec:Conclusions}
In this work we proposed and analyzed a coding strategy tailored for \gls{harq} protocol and aiming at the increase of the throughput for  transmission over block fading channel.  Unlike many heuristic coding schemes proposed  previously, our goal was to  address explicitly the issue of joint coding of many packets into the channel block of predefined length. With such a setup, the challenge is to optimize the coding rates for each packet which we do efficiently  assuming existence of a multi-bits feedback channel which transmit the outdated \gls{csi} experienced by the receiver.

The throughput of the resulting \gls{xp} is  compared to the conventional \gls{irharq} indicating that significant gains can be obtained using the proposed coding strategy. The gains are particularly notable in the range of high throughput, where the conventional \gls{harq} fails to offer any improvement with increasing number of transmission rounds. The proposed encoding scheme may be seen as a method to increase the throughput, or as a mean to diminish the memory requirements at the receiver; the price for the improvements is paid by a more complex joint encoding/decoding.

We also proposed an example of a practical implementation based on turbo codes. This example highlights the practical aspects of the proposed coding scheme, where the most important difficulties are i)~the need of tailoring the encoder to provide the jointly coded symbols with the best decoding performance, and ii)~the design of the simple decoder. Moreover, the real challenge is to leverage the possibility of adaptation to the outdated \gls{csi}. To do so, simple techniques for performance evaluation (\eg the \gls{per}) based on the expected \gls{csi}, must be used; such as, for example those studied in \cite{Latif13}.

\begin{appendices}

\section{Decoding conditions of \gls{xp}}\label{Sec:Decoding_Proof}

We outline the proof of the decoding conditions \eqref{dec.1} and \eqref{dec.2}, stated in the following  \lemref{lemma:DecodingConditions}. The HARQ-code refers to the encoding functions stated in \eqref{encoding.1} and \eqref{conventional.code} and the joint decoding of the pair $[\mfm_1, \mfm_2]$.


\begin{lemma}[Decoding conditions]\label{lemma:DecodingConditions}
For all $\varepsilon>0$, there exists $\bar{n}\in \mathbb{N}$ such that for all $n\geq \bar{n}$, there exists a HARQ-code $c^{\star} $ such that for all \gls{snr} realization $(\SNR_1,\SNR_2)$ that satisfy:
\begin{eqnarray}
R_1 +  R_2 &\leq& \mfI(X_1 ; Y_1 |\SNR_{1}) + \mfI(X_2 ; Y_2 |\SNR_{2}) -  \varepsilon, \label{eq:Rate2} \\
R_2 &\leq& \mfI(X_2 ; Y_2 |\SNR_{2}) -  \varepsilon , \label{eq:Rate1}
\end{eqnarray}
the error probability  is bounded by
\begin{eqnarray}
\PR{ [ \mfm_1, \mfm_2] \neq [ \hat{\mfm}_1, \hat{\mfm}_2]   \bigg|c^{\star}, \SNR_{1}, \SNR_{2} }\leq \varepsilon. 
\end{eqnarray}
\end{lemma}

\begin{proof}[Proof of Lemma \ref{lemma:DecodingConditions}]
We consider the random HARQ-code:
\begin{itemize}
\item[$\bullet$] \textit{Random codebook:} we generate $ 2^{\Ns \cdot R_1 } $ codewords $\bx_1$ and $2^{\Ns \cdot ( R_1 + R_2 )} $ codewords $\bx_2$, drawn from the uniform distribution over the constellation $\X$. 
\item[$\bullet$] \textit{Encoding function:} as explained in \secref{Sec:joint.codec}, the encoder starts by sending $\bx_1$ which corresponds to the packet (or \emph{message} in the language of information theory) $\mfm_1$. If the encoder receives a feedback $\nack_1$, it sends $\bx_2$ corresponding to the pair of messages $[\mfm_1 , \mfm_2]$. Otherwise a new transmission process starts. 
\item[$\bullet$] \textit{Decoding function:}  if the \gls{snr} realizations $(\SNR_1,\SNR_2)$  satisfy equations \eqref{eq:Rate1} and \eqref{eq:Rate2}, then the decoder finds a pair of messages $[\mfm_1,\mfm_2]$  such that the following sequences of symbols are jointly typical:
\begin{eqnarray}
\Big(\Phi_1[\mfm_1] , \by_1  \Big) \in A_{\varepsilon}^{{\star}{\Ns}},\; 
\Big(\Phi_2[\mfm_1, \mfm_2] , \by_2 \Big) \in A_{\varepsilon}^{{\star}{\Ns}}. \label{eq:JointTypical2}
\end{eqnarray}
\item[$\bullet$] \textit{Error} is declared when sequences  are not jointly typical. 
\end{itemize}

%

\textbf{Error events.} We define the following error events:  
\begin{footnotesize}
\begin{itemize}
\item[$\bullet$]$E_0=\bigg\{\Big(\Phi_1[\mfm_1] , \by_1  \Big) \notin A_{\varepsilon}^{{\star}{\Ns}}\bigg\}  \cup \bigg\{\Big(\Phi_2[\mfm_1, \mfm_2] , \by_2 \Big) \notin A_{\varepsilon}^{{\star}{\Ns}}\bigg\}$,\\
\item[$\bullet$]$E_{1} =\bigg\{\exists [\mfm_1',\mfm_2'] \neq [\mfm_1,\mfm_2],\text{ s.t. }\\
\qquad \qquad\qquad\qquad \qquad \Big\{ \Big(\Phi_1[\mfm_1'] , \by_1  \Big) \in  A_{\varepsilon}^{{\star}{\Ns}} \Big\} \cap \Big\{  \Big(\Phi_2[\mfm_1', \mfm_2'] , \by_2 \Big) \in A_{\varepsilon}^{{\star}{\Ns}} \Big\} \bigg\}$,\\
\item[$\bullet$]$E_{2} =\bigg\{\exists \mfm_1'\neq \mfm_1,\text{ s.t. }\\
\qquad \qquad \qquad\qquad\qquad\Big\{ \Big(\Phi_1[\mfm_1'] , \by_1  \Big) \in  A_{\varepsilon}^{{\star}{\Ns}} \Big\} \cap \Big\{  \Big(\Phi_2[\mfm_1', \mfm_2] , \by_2 \Big) \in A_{\varepsilon}^{{\star}{\Ns}} \Big\} \bigg\}$,\\
\item[$\bullet$]$E_{3} =\bigg\{\exists \mfm_2'\neq \mfm_2,\text{ s.t. }   \Big(\Phi_2[\mfm_1, \mfm_2'] , \by_2 \Big) \in A_{\varepsilon}^{{\star}{\Ns}}  \bigg\}$.
\end{itemize}
\end{footnotesize}

The properties of the typical sequences imply that, for $\Ns$ large enough, $\PR{E_0} \leq \varepsilon$, and the Packing Lemma \cite[p.~46]{ElGammalKim(book)11} implies that the probabilities of the events  $E_1$, $E_2$, $E_3$ are bounded by $\varepsilon$ if the following conditions are satisfied
\begin{eqnarray}
R_1 +  R_2 &\leq& \mfI(X_1 ; Y_1 |\SNR_{1}) + \mfI(X_2 ; Y_2 |\SNR_{2}) -  \varepsilon, \label{eq:Second1}\\
R_2 &\leq& \mfI(X_2 ; Y_2 |\SNR_{2}) -  \varepsilon \label{eq:Second3},\\
R_1  &\leq& \mfI(X_1 ; Y_1 |\SNR_{1}) + \mfI(X_2 ; Y_2 |\SNR_{2}) -  \varepsilon,  \label{eq:Second2}.
\end{eqnarray}
Since \eqref{eq:Second1}-\eqref{eq:Second3} are the hypothesis \eqref{eq:Rate2}-\eqref{eq:Rate1} of  \lemref{lemma:DecodingConditions}, there exists HARQ-code $c^{\star}$ with  small error probability. 


\end{proof}

\section{Optimization via \gls{mdp}}\label{Sec:OptimalAdaptation}

To obtain the \gls{mdp} formulation it is  convenient to replace packet-wise notation of \eqref{y.x.z} with a time-wise model 
\begin{align}\label{}
\by[n]=\sqrt{\SNR[n]}\bx[n]+\bz[n],
\end{align}
where $n$ is the index of the channel block.

At each time $n$, the \gls{harq} controller observes the \emph{state} $\mfs[n]$, and takes an \emph{action} $\mfa[n]=\pi(\mfa[n])$, according to the policy $\pi$. The transition probability matrix, $\bQ(\mfa)$, has the elements
\begin{align}\label{eq:trans_vl}
Q_{\mfs,\mfs'}(\mfa)\triangleq \displaystyle{\Pr\set{ \mfs[n+1]=\mfs'|\mfs[n]=\mfs, \mfa[n]=\mfa}},
 \end{align}
defining the probabilities of the system moving to the state $\mfs' \in \mcS$ at time $n+1$ conditioned on the system being in the state $\mfs \in \mcS$ at time $n$ and the controller taking the action $\mfa\in\mcA(\mfs)$, where $\mcA(\mfs)$ is the set of actions allowed in a state $\mfs$ and  $\underset{\mfs \in \mcS}{\bigcup}\mcA(\mfs)=\mcA$. In our case, the actions are the coding rates, which we assume may take any positive value, and thus $\mcA(\mfs)=\Real_+$.

A \emph{policy} $\pi$ is defined as a mapping $\pi:\mcS\mapsto\mcA$ between the state space, $\mcS$, and the action space, $\mcA$. We aim at finding a policy $\pi$ which maximizes the long-term average throughput
\begin{align}\label{throughput.t}
\eta(\pi) &= \lim_{N \rightarrow \infty}\frac{1}{N}\sum_{n=1}^N \Ex\big[ \mfR(\mfs[n],\pi(\mfs[n])) \big],
 \end{align}
where $\mfR(\mfs,\mfa)$ is the average reward obtained when taking action $\mfa$ in the state $\mfs$ and the  expectations are taken with respect to the random states $\mfs[n]$. In our case the reward is the number of decoded bits normalized by the duration of the channel block, $\Ns$. 

The optimal policy thus solves the following problem:
\begin{align}\label{op.throughput.t}
\hat{\eta}^{\xp}_{\kmax}=\max_{\pi(\cd)}\eta(\pi)
 \end{align}
 and may be found solving the Bellman equations \cite[Prop.~4.2.1]{Bertsekas07_book}
\begin{align}\label{bellman.eq}
\hat{\eta}^{\xp}_{\kmax}+h(\mfs)&=\max_{\mfa \in \mcA(\mfs)}\left[ \mfR(\mfs,\mfa) + \sum_{\mfs'\in\mcS} 
Q_{\mfs,\mfs'}(\mfa) h(\mfs')\right],\quad \forall \mfs\in\mcS,
\end{align}
where $h(\mfs)$ is a difference reward associated with the state. To calculate the optimal $\hat{\eta}^{\xp}_{\kmax}$, we use here the policy iteration algorithm whose details may be found in \cite[Sec.~4.4.1]{Bertsekas07_book} and which guarantees to reach the solution after a finite number of iterations. 

The unique optimal throughput $\hat{\eta}^{\xp}_{\kmax}$ exists and is independent of the initial state, $\mfs[0]$ if, for any state $\mfs'[t]\in\mcS$, we can find a policy, which starting with arbitrary state $\mfs[0]$ reaches the state $\mfs'[t]$ in a finite time $t<\infty$, with non-zero probability \cite[Prop.~4.2.6 and Prop.~4.2.4]{Bertsekas07_book}. For our problems, finding such a policy is indeed possible, proof of which we skip for sake of brevity.

In order to define the state space and the average reward, we deal separately with the truncated and  persistent \gls{xp} but in both cases we must track the accumulated rate, $\Rsig[n]$ (it defines the reward, $\mfR(\mfs,\mfa)$), and the accumulated \gls{mi}, $\Isig[n]$ (it defines the matrix $\bQ$). Thus these two variables must enter the definition of the state, $\mfs[n]$. 

\subsection{Persistent HARQ}

For the persistent \gls{xp}, the state can be defined as a pair
\begin{align}\label{state.persistent}
\mfs[n]\triangleq(\Isig[n],\Rsig[n]),
\end{align}
and the transition to the state at time $n+1$ is defined as
\begin{align}\label{eq:dynamic.pers}
\mfs[n+1]=
\begin{cases}
\big(\Isig[n]+I[n],\Rsig[n]+\R[n]\big), \\
 \qquad~\quad \text{if}  \quad  \Rsig[n]+\R[n]\geq \Isig[n]+I[n]\\
\big(0, 0\big), \quad \text{otherwise}.
\end{cases}.
\end{align}
A non-zero reward is obtained only by terminating the \gls{harq} cycle, \ie moving to the state  $\mfs[n+1]=(0,0)$,
\begin{align}\label{reward.persistent}
\mfR(\mfs[n],\mfa)=&\big(\Rsig[n]+\mfa\big)\ccdf_I (\Rsig[n]-\Isig[n]+\mfa),
\end{align}
where $\ccdf_I(x)\triangleq 1-\cdf_I(x)$ and $\cdf_I(x)$ is the \gls{cdf} of $I$.

\subsection{Truncated HARQ}
In the truncated \gls{harq}, a new \gls{harq} cycle starts also if the maximum number of allowed rounds is attained (even if the message is not decoded correctly). Thus i)~the index of the transmission round, $\mfk$, must enter the defining of the state, ii)~we need to make a distinction between the decoding success/failure of the last round. We thus define the state as 
 \begin{align}\label{state.persistent}
\mfs[n]\triangleq(\Isig[n],\Rsig[n],\mfk[n],\mfM[n]),
\end{align}
where $\mfk[n]$ and $\mfM[n]\in\set{\ack,\nack}$ are respectively, the number of rounds and the decoding result after the transmission in block $n$. The system dynamic is described as follows:
\begin{align}\nonumber
\mfs[n+1]=
\begin{cases}
\big(0, 0,0,\ack\big), \quad ~~\text{if} \quad \Eharqa[n] \\
\big(0, 0,0,\nack\big), \quad \text{if} \quad \Eharqn[n] \\
\big(\Isig[n]+I[n],\Rsig[n]+\R[n],\mfk[n]+1,\nack\big), \\
 \qquad~\qquad \text{otherwise} 
\end{cases}
\end{align}
where 
\begin{align}\nonumber
\Eharqa[n]&\triangleq\set{\Rsig[n]+\R[n]\leq \Isig[n]+I[n]}\\
\nonumber
\Eharqn[n]&\triangleq\set{\Rsig[n]+\R[n]> \Isig[n]+I[n]~\wedge~\mfk[n]+1=\kmax}
\end{align}
are respectively, the conditions indicating a successful decoding and a decoding failure at the end of the \gls{harq} cycle. 

Thus, the state space is defined as: $\mcS=\Real_{+}\times\Real_{+}\times\set{0,1,\ld,\kmax-1}\times\set{\ack,\nack}$ and the reward is defined by \eqref{reward.persistent}.

\section{Optimal MDP for $K=2$}\label{Sec:OptimalAdaptation.Analytic}

Knowing the rate of the first transmission, $\R_{1}$, the optimization problem \eqref{op.throughput.t} may be solved analytically for $\kmax=2$ using \eqref{eq.mm.ad}
\begin{align}\label{op.throughput.definition.RR.2} 
&\displaystyle{\hat{\eta}^\xp_2=\max_{ \R_{2}(I_{1})} \frac{\Ex\Big[\R_{1}\IND{I_{1}\ge \R_{1}}\Big]}{1+f_{1}}+}\nonumber \\ 
&\displaystyle{\frac{\Ex\Big[(\R_{1}+ \R_{2}(I_{1}))\IND{I_{1}\le \R_{1} \wedge \Isig_2\ge \R_{1}+\R_{2}(I_{1})}\Big]}{1+f_{1}}}.
\end{align}

Since $f_{1}$ is independent of $\R_{2}(\cd)$, solving \eqref{op.throughput.definition.RR.2} is equivalent to
finding, for each value of $I_{1}<\R_{1}$, the optimal $\R_{2}(\cd)$ as follows
\begin{align}\label{} 
\R_{2}(I_{1})=\argmax_{R}~(\R_{1}+ R) \cd \cdf_{I_{2}}^{\tr{c}}(\R_{1}+R-I_{1}).
\end{align}
which is a  one-dimension optimization problem, that can be solved analytically, provided $\cdf_{I_{2}}^{\tr{c}}(\cd)$ is known.

In the case of Gaussian codebook, \ie when the \gls{mi} is given by $I_{k}=\log_{2}(1+\SNR_{k})$, the optimal rate adaptation policy is given by the following closed-form
\begin{align}\label{eq:optimal.analitical.policy} 
\R_{2}(I_{1})=\max\big(0,\frac{W(2^{I_{1}} \SNRav )}{\log(2)}-R_{1}\big),
\end{align}
where $W(.)$ is Lambert $W$ function defined as the solution of $x=W(x)\e^{W(x)}$.

\end{appendices}


\bibliographystyle{IEEEtran}

\balance

\end{document}